%% file: main.tex
\documentclass{IOS-Book-Article}

\usepackage{mathptmx}
\usepackage{amsmath,amssymb,amsthm}
\usepackage{thmtools,thm-restate}
\usepackage{url}
\usepackage{tikz}
\usetikzlibrary{arrows.meta,calc,fit,positioning}

\def\hb{\hbox to 11.5 cm{}}

\makeatletter
\g@addto@macro\normalsize{%
  \abovedisplayskip 6pt plus 2pt minus 3pt
  \belowdisplayskip \abovedisplayskip
  \abovedisplayshortskip 2pt plus 2pt
  \belowdisplayshortskip 4pt plus 2pt minus 2pt}
\makeatother
\newcommand{\ia}{\textit{i}}
\newcommand{\ib}{\textit{ii}}
\newcommand{\ic}{\textit{iii}}

\newcommand{\calA}{\mathcal{A}}
\newcommand{\calF}{\mathcal{F}}
\newcommand{\calS}{\mathcal{S}}
\newcommand{\Alice}{\mathit{Alice}}
\newcommand{\Bob}{\mathit{Bob}}

\theoremstyle{definition}
\newtheorem{definition}{Definition}
\theoremstyle{plain}
\newtheorem{lemma}[definition]{Lemma}
\newtheorem{proposition}[definition]{Proposition}
\newtheorem{theorem}[definition]{Theorem}
\newtheorem{corollary}[definition]{Corollary}

\newif\iffull
 \fulltrue
\newcommand{\full}[2]{\iffull#1\else#2\fi}

\full{\newcommand{\mypara}[1]{\smallskip\noindent\textbf{#1.}}}{\newcommand{\mypara}[1]{\noindent\textbf{#1.}}}

\makeatletter
\def\thm@space@setup{\thm@preskip=\full{5pt plus 1pt minus 1pt}{3pt plus 1pt minus 1pt} \thm@postskip=\thm@preskip}
\renewcommand\section{\@startsection{section}{1}{\z@}{\full{-\bigskipamount}{-9pt plus -2pt minus -2pt}}{\full{\medskipamount}{4pt plus 1pt minus 1pt}}{\normalsize\bfseries\nohyphen\raggedright}}
\renewcommand\subsection{\@startsection{subsection}{2}{\z@}{\full{-\medskipamount}{-7pt plus -2pt minus -1pt}}{\full{\medskipamount}{3pt plus 1pt minus 1pt}}{\normalsize\itshape\nohyphen\raggedright}}
\renewenvironment{proof}[1][\proofname]{\par\pushQED{\qed}\normalfont\topsep\full{6pt plus 6pt}{3pt plus 2pt minus 1pt}\relax\trivlist\item[\hskip\labelsep\itshape#1\@addpunct{.}]\ignorespaces}{\popQED\endtrivlist\@endpefalse}
\def\@listI{\leftmargin\leftmargini \parsep 0pt plus 1pt \topsep \full{9pt plus 2pt minus 2pt}{4pt plus 1pt minus 1pt} \partopsep 1pt \itemsep 1pt plus .5pt minus .5pt}
\let\@listi\@listI \@listi
\makeatother

\makeatletter\def\@abstract@width{\textwidth}\makeatother

\begin{document}

\pagestyle{plain}
\begin{frontmatter}

\title{\Large Formalising and Implementing Grassroots Social Contracts:\\
From Legal Text to Grassroots Platforms\full{\\ (Full Version)}{}}

\author[A]{\fnms{James} \snm{Golike}},
\author[A]{\fnms{Andy} \snm{Lewis-Pye}} and
\author[A,B]{\fnms{Ehud} \snm{Shapiro}}

\address[A]{London School of Economics and Political Science}
\address[B]{Weizmann Institute of Science}

\begin{abstract}
Almost two centuries ago Pierre-Joseph Proudhon envisioned a social contract that is (1) an agreement of man with man; (2) reciprocal; (3) imposing no obligation upon the parties except that which results from their personal promise; (4) subject to no external authority; (5) freely accepted and signed by all the participants; (6) of the nature of a contract of exchange; and a few other conditions.  He also envisioned a property of social contracts, analogous to a property of digital platforms we term \emph{grassroots}: 
that (\ia) one could ``make a contract with all, as \ldots\ with some''; digitally, that a grassroots platform can have multiple instances, and (\ib) ``each group of citizens \ldots\ formed by a like contract \ldots\ could thereafter, and always by a similar contract, agree with every and all other groups''; digitally, that independent platform instances may coalesce by mutual consent, possibly into a single global platform instance. \iffull We proposed grassroots platforms owned, operated and governed by their participants as a foundation for an equitable and democratic digital realm.\fi

Here we define \emph{grassroots social contracts} as social contracts meeting Proudhon's six conditions above and that (7)  people are free to deal with each other; and (8) there is no external register of people, so people can come to know each other either directly or through people they already know.

We show that a grassroots social contract can be transformed into a working smartphone-based grassroots platform through an \emph{abstraction cascade}, starting from converting the social contract text to formal act schemas, verified syntactically to be grassroots, and then to volition-guarded multiagent atomic transactions, hitherto the most abstract formalism used to specify grassroots platforms. \iffull The abstraction cascade from there to a working app has been described elsewhere.\fi

\emph{Act schemas} are a formal language for the acts the contract describes, each naming the parties' roles, which of them must will the act, and its precondition and effect at each role, with contract parties determined upon signature.  Any contract written in this language meets the eight conditions above, provided it is \emph{syntactically grassroots}, satisfying\full{: (\ia) \emph{Introduction}, that some act can introduce any two willing parties, following which the state of each stores an identifier of the other; (\ib) \emph{Provenance}, that the state of a party stores an identifier of a person only by an act with that person, or with a party whose state stores that identifier; and (\ic) \emph{Volition}, that an act requires the consent of all its parties, unless the undirected graph formed by identifiers stored by the parties is connected.}{ three decidable conditions, \emph{Introduction}, \emph{Provenance} and \emph{Volition}.}

We prove that the protocol realising a syntactically grassroots contract is grassroots, and volitionally so. \full{To do so, we first prove that \emph{politeness} --- any two parties can always come to interact, and no act joins two instances unless all its participants will it --- is sufficient for the protocol over a set of volition-guarded multiagent atomic transactions to be grassroots, and volitionally so: two groups coalesce only when a member of each wills it.  
We then prove that a syntactically grassroots social contract compiles into a polite set of volition-guarded transactions.}{Politeness of a set of volition-guarded transactions, a second theorem, makes the protocol over them volitionally grassroots, and a syntactically grassroots social contract compiles into a polite set.}  

We distinguish three legal types of clauses of a grassroots social contract: A breach of an \emph{enforced} clause is impossible with a correct implementation of the contract; a breach of an \emph{attested} or \emph{undertaken} clause can be taken to court, with non-repudiably signed evidence of the breach in the case of an attested clause.  We illustrate the three legal types with \full{two grassroots social contracts, the grassroots social graph carrying signed items and grassroots currencies}{the grassroots social graph carrying signed items, and in the full version of this paper also grassroots currencies}.

To join a grassroots platform, a person signs a one-sided undertaking to keep its grassroots social contract towards all its other participants, and presents it to every participant it interacts with, as enforced by the code realising the social contract or the platform hosting it.  Hence by transitivity every participant has signed it.   The code is certified by its creator and by its compiler, and each party undertakes to run the certified code.
\end{abstract}

\end{frontmatter}

\input{sections/01-introduction}
\input{sections/02-contracts}
\input{sections/03-schemas}
\input{sections/04-transactions}
\input{sections/05-compilation}
\input{sections/06-polite}
\input{sections/07-conditions}
\input{sections/08-modalities}
\input{sections/09-signing}
\input{sections/10-related-work}
\input{sections/11-conclusion}

\bibliographystyle{vancouver}
\bibliography{bib}

\iffull
\appendix
\input{sections/12-appendix}
\fi

\end{document}

%% file: sections/01-introduction.tex
\section{Introduction}\label{sec:intro}

Almost two centuries ago Pierre-Joseph Proudhon envisioned a social contract~\cite[Fourth Study]{proudhon2004general} \full{that is (1) an agreement of man with man; (2) reciprocal; (3) imposing no obligation upon the parties except that which results from their personal promise; (4) subject to no external authority; (5) freely accepted and signed by all the participants; (6) of the nature of a contract of exchange; and a few other conditions, among them that it include all citizens with their interests and that it increase the well-being and liberty of every citizen.  \iffull In his words, the social contract is ``an agreement of man with man''; ``essentially reciprocal: it imposes no obligation upon the parties, except that which results from their personal promise of reciprocal delivery: it is not subject to any external authority: it alone forms the law between the parties: it awaits their initiative for its execution''; ``freely discussed, individually accepted, signed with their own hands, by all the participants''; and ``of the nature of a contract of exchange: not only does it leave the party free, it adds to his liberty; not only does it leave him all his goods, it adds to his property; it prescribes no labor; it bears only upon exchange''.\fi{}  He also envisioned a property of social contracts~\cite[Sixth Study]{proudhon2004general}, analogous to a property of digital platforms we term \emph{grassroots}~\cite{shapiro2023grassrootsBA,shapiro2024grassroots,shapiro2025characterising}:  that (\ia) one could ``make a contract with all, as \ldots\ with some''; digitally, that a grassroots platform can have multiple instances, and (\ib) ``each group of citizens \ldots\ formed by a like contract \ldots\ could thereafter, and always by a similar contract, agree with every and all other groups''; digitally, that independent platform instances may coalesce by mutual consent, possibly into a single global platform instance.  \iffull In full: ``if I could make a contract with all, as I can with some; if all could renew it among themselves, if each group of citizens, as a town, county, province, corporation, company, \&c., formed by a like contract, and considered as a moral person, could thereafter, and always by a similar contract, agree with every and all other groups, it would be the same as if my own will were multiplied to infinity.  I should be sure that the law thus made on all questions in the Republic, from millions of different initiatives, would never be anything but my law; and if this new order of things were called government, it would be my government.''\fi{} \iffull We proposed grassroots platforms owned, operated and governed by their participants as a foundation for an equitable and democratic digital realm~\cite{shapiro2024grassroots}.\fi}{meeting conditions (1)--(6) of the abstract, and the property we term grassroots~\cite[Sixth Study]{proudhon2004general}.}

A \emph{digital social contract} is one realised by software the parties run~\cite{cardelli2020digital}.  We call the mathematical object a \emph{protocol} and the software realising it a \emph{platform}.  We say a platform \emph{realises} a contract, and reserve \emph{implements} for the relation between two adjacent abstractions of the abstraction cascade~\cite{shapiro2026volition}.  A protocol is \emph{grassroots} when its instances are independent of one another and of every global resource but the net, and can always coalesce, possibly into a single global instance~\cite{shapiro2023grassrootsBA,shapiro2024grassroots,shapiro2025characterising}, and a \emph{grassroots platform} is the software realising a grassroots protocol.  The grassroots notion is a formal property of a protocol, recalled in Section~\ref{sec:transactions}.

\iffull
Here we define \emph{grassroots social contracts} as social contracts meeting Proudhon's six conditions above and that (7)  people are free to deal with each other; and (8) there is no external register of people, so people can come to know each other either directly or through people they already know.  \iffull Our first theorem is that such a contract has the grassroots property: the acts its clauses describe determine a grassroots platform, in which any two parties may come into relation without requiring anyone's permission, two communities that adopted the contract independently coexist and may later merge, and every participant in any act is a party.  No registry issues identifiers, no log orders acts, and no authority timestamps them.\fi
\fi

\iffull
A social contract is a contract among those who adopt it.  Participating in a grassroots platform requires adopting the grassroots social contract that defines it, as participating in a global platform requires adopting an End-User Licence Agreement (EULA).  A grassroots social contract is mutual and symmetric among all participants; an EULA is between each person and the platform operator.  A friendship or a mutual credit line is an act under the contract between particular parties.

\fi
\full{We show contracts for a social network carrying signed speech acts~\cite{shapiro2023gsn,golike2026digital} and for a currency of personal coins~\cite{shapiro2024gc} that are grassroots.}{We show that a contract for a social network carrying signed speech acts~\cite{shapiro2023gsn,golike2026digital} is grassroots, and in the full version of this paper~\cite{shapiro2026formalising} a contract for a currency of personal coins~\cite{shapiro2024gc}.}  \iffull Not every social contract is.  Examples that are not include a contract without coalescences, not allowing two strangers to come into relation; a contract with named third-party control, e.g.\ where introduction requires a token only a third party can make; and a contract with an act that changes a stranger's state without their will and requires no prior relation between the two; each fails one of the conditions.\fi

\iffull
We show that a grassroots social contract can be transformed into a working smartphone-based grassroots platform through an \emph{abstraction cascade}~\cite{shapiro2026volition}, starting from converting the social contract text to formal act schemas, verified syntactically to be grassroots, and then to volition-guarded multiagent atomic transactions~\cite{lewis2026volitional}, hitherto the most abstract formalism used to specify grassroots platforms. \iffull The abstraction cascade from there to a working app has been described elsewhere~\cite{shapiro2026volition}.\fi
\fi
\full{}{The cascade continues through Grassroots Logic Programs~\cite{shapiro2025glp,shapiro2026implementing,shapiro2026types} to Dart~\cite{dart2024}, by the Human--Science--AI methodology~\cite{shapiro2026volition}, at the antipode of vibe coding~\cite{sapkota2025vibe} and beside spec-driven development~\cite{fowler2025sdd,github2025speckit}; a child-safe social network~\cite{shapiro2026cssn} is specified the same way.}

\iffull
\iffull A grassroots social contract has two artefacts: a legal text, and a protocol realising it.  A clause of the text describes an \emph{act}, in the sense of a juridical act: a declaration by one or more parties intended to create a legal effect.\fi{}  \emph{Act schemas} are a formal language for the acts the contract describes, each naming the parties' roles, which of them must will the act, and its precondition and effect at each role, with contract parties determined upon signature.  Any contract written in this language meets the eight conditions above, provided it is \emph{syntactically grassroots}, satisfying\full{: (\ia) \emph{Introduction}, that some act can introduce any two willing parties, following which the state of each stores an identifier of the other, as of a friend, the issuer of a coin it holds, or the author of an item it was sent, and we say that it \emph{records} the other; (\ib) \emph{Provenance}, that the state of a party stores an identifier of a person only by an act with that person, or with a party whose state stores that identifier; and (\ic) \emph{Volition}, that an act requires the consent of all its parties, unless the undirected graph formed by identifiers stored by the parties is connected.}{ three decidable conditions, \emph{Introduction}, \emph{Provenance} and \emph{Volition}.}  Conditions (1), (2) and (4) hold of every contract written as act schemas, since every participant in an act is a party and the schemas single out no one, and so does (5), since nobody is bound who has not adopted the contract (\full{Section~\ref{sec:proudhon}}{Propositions~\ref{prop:consent} and~\ref{prop:anonymity}}); (3) is Volition, (6) and (7) are Introduction, and (8) is Provenance (Section~\ref{sec:informal-conditions}).  The three are decidable, and a checker for them accompanies this paper, at \url{https://github.com/EShapiro2/GLP/tree/main/programs/jurix}; it certifies both of the contracts we carry through.  The compilation into transactions is defined for contracts meeting them.
\fi

\iffull
We prove that the protocol realising a syntactically grassroots contract is grassroots, and volitionally so (Theorem~\ref{thm:grassroots}). \full{To do so, we first prove that \emph{politeness} --- any two parties can always come to interact, and no act joins two instances unless all its participants will it --- is sufficient for the protocol over a set of volition-guarded multiagent atomic transactions to be grassroots, and volitionally so: two groups coalesce only when a member of each wills it (Theorem~\ref{thm:polite}).   We then prove that a syntactically grassroots social contract compiles into a polite set of volition-guarded transactions (Section~\ref{sec:conditions}).}{Politeness of a set of volition-guarded transactions, a second theorem (Theorem~\ref{thm:polite}), makes the protocol over them volitionally grassroots, and a syntactically grassroots social contract compiles into a polite set (Section~\ref{sec:conditions}).}
\fi

\iffull
We distinguish three legal types of clauses of a grassroots social contract (Section~\ref{sec:modalities}), by whether a party can breach a clause and what record of it the contract leaves in another party's state: a clause is \emph{enforced}, if performing it is an act of the contract and no act of the contract can breach it; \emph{attested}, if a party may breach it, with the contract recording the breach; or \emph{undertaken}, if the contract records who gave the undertaking, but performing the undertaking is carried out outside the contract. A breach of an \emph{enforced} clause is impossible with a correct implementation of the contract; a breach of an \emph{attested} or \emph{undertaken} clause can be taken to court, with non-repudiably signed evidence of the breach in the case of an attested clause.  We illustrate the three legal types with \full{two grassroots social contracts, the grassroots social graph carrying signed items~\cite{shapiro2026gsg,shapiro2023gsn,golike2026digital} and grassroots currencies~\cite{shapiro2024gc}}{the grassroots social graph carrying signed items~\cite{shapiro2026gsg,shapiro2023gsn,golike2026digital}, and in the full version of this paper~\cite{shapiro2026formalising} also grassroots currencies~\cite{shapiro2024gc}}.  The compilation and the conditions are algorithmic, so the only informal step in the chain is the first: rendering a legal text as schemas.
\fi

\iffull
To join a grassroots platform, a person signs a one-sided undertaking to keep its grassroots social contract towards all its other participants, and presents it to every participant it interacts with, as enforced by the code realising the social contract or the platform hosting it.  Hence by transitivity every participant has signed it.  Enforced clauses are enforced by the code the parties run.  The code is certified by its creator and by its compiler, and each party undertakes to run the certified code.
\fi

\iffull
We carry two contracts through, a social graph carrying signed items and a currency of personal coins, and certify both from their schemas.  A coin can be handed on, so a holding may record a person its holder has never transacted with; the certification then rests on traceable provenance, and conservation of money is its special case.

\fi
\iffull
\mypara{Three languages} The paper employs three languages: natural language for the social contract; act schemas; and volition-guarded atomic transactions.  Rendering the first as the second is by judgement; compiling the second into the third is algorithmic; and the conditions of the main theorem, that a syntactically grassroots social contract is grassroots, are verified in the second.  Clause~1 of the social graph contract (Section~\ref{sec:sg}) reads ``Alice and Bob may jointly decide to become friends, if neither is the friend of the other''.  Its schema (Section~\ref{sec:sg-schemas}) names the act's two roles, marks each with $?$ as having to will it, and at each role forbids the atom that role gains and adds it:
\[
\mathit{befriend}(\Alice?,\Bob?) : \ \ \neg\mathit{friend}(\Bob),\,+\mathit{friend}(\Bob) \ \ \ \neg\mathit{friend}(\Alice),\,+\mathit{friend}(\Alice).
\]
Compiling it (Section~\ref{sec:compiled-examples}) yields the volition-guarded transaction $(d \rightarrow d',\{\Alice,\Bob\})$, for every $d$ with $\mathit{friend}(\Bob) \notin d_{\Alice}$ and $\mathit{friend}(\Alice) \notin d_{\Bob}$, where $d'_{\Alice} = d_{\Alice} \uplus \{\mathit{friend}(\Bob)\}$ and $d'_{\Bob} = d_{\Bob} \uplus \{\mathit{friend}(\Alice)\}$.  The rest of the paper concerns the middle line.

\fi
\iffull
\mypara{The abstraction cascade} The passage from a legal text to a running platform goes through the \emph{abstraction cascade}~\cite{shapiro2026volition} set out in Figure~\ref{fig:cascade}, from the legal text at the top to an app on each party's own phone at the bottom.  This paper provides (1)$\rightarrow$(3).  The steps below have been carried out by the \emph{Human--Science--AI methodology}, in which AI derives code from a paper written for peer review and the paper is repaired wherever coding finds a gap in it: the grassroots social graph~\cite{shapiro2026gsg}, the grassroots social network~\cite{shapiro2023gsn} and grassroots currencies~\cite{shapiro2026bonds} are each a volition-guarded Grassroots Logic Program generated from its volition-guarded transactions, and one Dart bridge renders the three as panels of a single app, deployed on a physical smartphone~\cite{shapiro2026volition}.  A child-safe social network~\cite{shapiro2026cssn} is specified the same way.

\fi
\iffull
\mypara{The Human--Science--AI methodology} The formal objects of the cascade are its abstractions and the implementations among them, each constructed, maintained, evolved and proved correct with AI, in tandem with the paper describing it.  The methodology has two kinds of artefact: papers written for competitive peer review, and code.  The papers serve as an evolving and cooperatively-constructed interface between us humans and AI, recording our shared understanding of the subject matter, and are the source of authority for coding.  Papers are written by people and reviewed, revised and extended with the help of AI; the code, with few exceptions, is written by AI; and both evolve in the process of harmonising them.  AI codes down the abstraction cascade, from a paper destined for peer review to tested code, and every gap coding finds is repaired in the paper before it is repaired in the code.  It is at the antipode of ``vibe coding'', where a person describes the intended behaviour informally and AI guesses the rest~\cite{sapkota2025vibe}.  Spec-driven development makes the specification the primary artefact, maintained by people while the code is regenerated from it~\cite{fowler2025sdd,github2025speckit}; here that artefact is a scientific paper, and the feedback from coding revises it.  Nine components were carried to tested code in a single year, each in tandem with the paper describing it.

\fi
\iffull
\mypara{Relation to contract specification languages} Formalising legal text has a large literature, on contracts as automata, as defeasible deontic rules and as domain-specific languages, and on statute as code, and Section~\ref{sec:related} reviews it.  Two languages formalise legal contracts in a way close enough to ours to say plainly how they differ.  \emph{Symboleo}~\cite{sharifi2020symboleo,parvizimosaed2022symboleo} specifies a contract as obligations and powers with statechart lifetimes, for verification and for monitoring; \emph{Stipula}~\cite{crafa2022stipula} is a domain-specific language whose contracts are state machines over parties, fields and linear assets, with an agreement operator and a formal semantics.  Both bind their parties at contract formation and both keep one global contract state.  Our question is about the family of systems a contract induces over every finite set of people, and about what happens when two groups that formed independently meet; it cannot be posed in either language, and each requires a global resource that closure forbids.  Section~\ref{sec:related} takes this up, together with the literature on norms and on contract languages more broadly.

\fi
\iffull
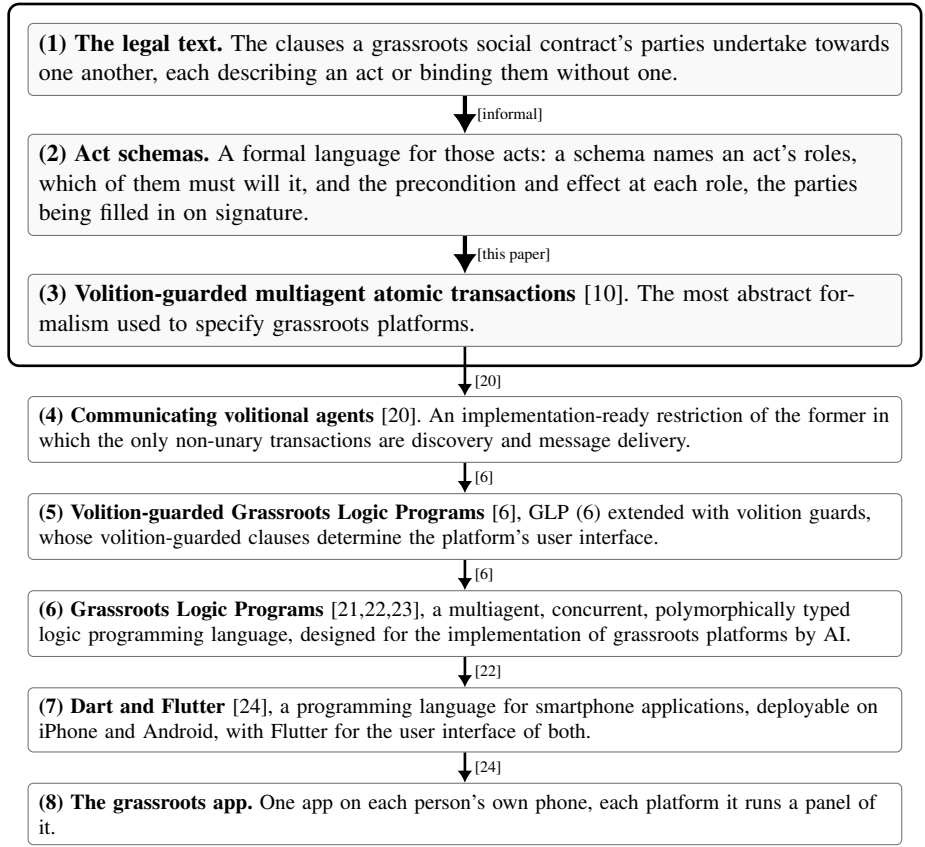
\begin{figure}[!t]
\centering
\begin{tikzpicture}[
  layer/.style={draw=black!55, rounded corners=2pt, fill=gray!5,
                text width=\dimexpr\linewidth-1.1cm\relax, inner sep=4pt,
                align=left, font=\small},
  low/.style={layer, fill=none, font=\scriptsize},
  arrbig/.style={-{Latex[length=2.6mm,width=3.2mm]}, line width=1.6pt},
  arrsmall/.style={-{Latex[length=1.6mm,width=2mm]}, line width=1pt},
  lab/.style={font=\tiny, inner sep=1pt, midway, right=1mm},
  labs/.style={lab, right=0.8mm}]
  \node[layer] (c1) {\textbf{(1) The legal text.}  The clauses a grassroots social contract's parties undertake towards one another, each describing an act or binding them without one.};
  \node[layer, below=5mm of c1] (c2) {\textbf{(2) Act schemas.}  A formal language for those acts: a schema names an act's roles, which of them must will it, and the precondition and effect at each role, the parties being filled in on signature.};
  \node[layer, below=5mm of c2] (c3) {\textbf{(3) Volition-guarded multiagent atomic transactions}~\cite{lewis2026volitional}.  The most abstract formalism used to specify grassroots platforms.};
  \node[low, below=6.5mm of c3] (c4) {\textbf{(4) Communicating volitional agents}~\cite{keidar2026ivmat}.  An implementation-ready restriction of the former in which the only non-unary transactions are discovery and message delivery.};
  \node[low, below=4mm of c4] (c5) {\textbf{(5) Volition-guarded Grassroots Logic Programs}~\cite{shapiro2026volition}, GLP~(6) extended with volition guards, whose volition-guarded clauses determine the platform's user interface.};
  \node[low, below=4mm of c5] (c6) {\textbf{(6) Grassroots Logic Programs}~\cite{shapiro2025glp,shapiro2026implementing,shapiro2026types}, a multiagent, concurrent, polymorphically typed logic programming language, designed for the implementation of grassroots platforms by AI.};
  \node[low, below=4mm of c6] (c7) {\textbf{(7) Dart and Flutter}~\cite{dart2024}, a programming language for smartphone applications, deployable on iPhone and Android, with Flutter for the user interface of both.};
  \node[low, below=4mm of c7] (c8) {\textbf{(8) The grassroots app.}  One app on each person's own phone, each platform it runs a panel of it.};
  \draw[arrbig] (c1) -- node[lab] {[informal]} (c2);
  \draw[arrbig] (c2) -- node[lab] {[this paper]} (c3);
  \draw[arrsmall] (c3) -- node[labs,pos=0.72] {\cite{keidar2026ivmat}} (c4);
  \draw[arrsmall] (c4) -- node[labs] {\cite{shapiro2026volition}} (c5);
  \draw[arrsmall] (c5) -- node[labs] {\cite{shapiro2026volition}} (c6);
  \draw[arrsmall] (c6) -- node[labs] {\cite{shapiro2026implementing}} (c7);
  \draw[arrsmall] (c7) -- node[labs] {\cite{dart2024}} (c8);
  \node[draw,line width=1pt,rounded corners=4pt,fit=(c1)(c2)(c3),inner sep=2.5mm] {};
\end{tikzpicture}
\caption{The abstraction cascade, from a legal text to a running platform.  The framed (1)--(3) are this paper.  Each arrow but the first is an implementation of the abstraction above it by the one below, carrying the reference in which that implementation is defined; the first, rendering a legal text as act schemas, is the one step left to judgement.  A condition on the transactions of~(3), politeness, is sufficient for the platform they specify to be a grassroots platform, and conditions on the schemas of~(2) are sufficient in turn for politeness.}
\label{fig:cascade}
\end{figure}

\fi
\mypara{Structure} Section~\ref{sec:contracts} presents the \full{two contracts}{contract} and Section~\ref{sec:schemas} the syntax of act schemas, with \full{the schemas of both}{its schemas}.  Section~\ref{sec:transactions} recalls what is needed of~\cite{lewis2026volitional} and Section~\ref{sec:compilation} compiles the schemas into its transactions.  Section~\ref{sec:polite} defines politeness and proves it sufficient for a protocol to be volitionally grassroots, and Section~\ref{sec:conditions} states the conditions on the schemas and certifies \full{both contracts}{the contract} from them.  Section~\ref{sec:modalities} takes up the clauses and their modalities, Section~\ref{sec:signing} how a person enters the contract and what the code it runs warrants, and Section~\ref{sec:related} reviews related work.  \full{Appendix~\ref{app:framework} states the framework of~\cite{lewis2026volitional} formally, citing its proofs.}{Proofs omitted here are in the full version of this paper~\cite{shapiro2026formalising}.}

%% file: sections/02-contracts.tex
\section{\full{Two Grassroots Social Contracts}{A Grassroots Social Contract}}\label{sec:contracts}

A grassroots social contract is a finite list of clauses its parties undertake towards one another.  A \emph{party} is a person, natural or legal, with the gender-neutral pronoun `they'; Alice and Bob are variables naming distinct parties.  A clause either describes an \emph{act}, which the protocol realising the contract carries out, or binds the parties without describing one.  \iffull The two contracts are examples: the syntax of Section~\ref{sec:schemas}, the compilation of Section~\ref{sec:compilation} and the conditions of Section~\ref{sec:conditions} apply to any contract written the same way, and these two are the ones the paper certifies.\fi

\iffull
Two contracts are carried through the paper: a social graph carrying signed items, and a currency of personal coins.  Their clauses are stated here in prose, as a contract states them.  The sections that follow give the acts a formal syntax, compile the schemas into transactions, and certify both contracts from the schemas.

\fi
\subsection{The Social Graph with Signed Items}\label{sec:sg}

The parties to this contract are the members of a social network, and its clauses are five.  The first two are the clauses of the grassroots social graph~\cite{shapiro2026gsg,shapiro2023gsn}, the platform the rest extends: they say who may become whose friend.  The next two carry signed items over it~\cite{golike2026digital}, and the fifth is added here as a clause that describes no act; Section~\ref{sec:modalities} says what a realisation does with such a clause.

\begin{enumerate}
\item \emph{Friendship.}  Alice and Bob may jointly decide to become friends, if neither is the friend of the other.  Each is thereafter a friend of the other.
\item \emph{Ending a friendship.}  Alice may decide to end their friendship with Bob, and it then ends for both.
\item \emph{Authorship.}  Alice may decide to sign an item of their own.
\item \emph{Forwarding.}  Alice may decide to forward an item they hold to their friend Bob, and warrants that they know personally the party they received the item from.
\item \emph{Confidence.}  Alice shall disclose an item they hold only by forwarding it.
\end{enumerate}

\noindent The first four clauses describe acts and determine the contract's transactions.  The warranty in the fourth describes no act of its own, and neither does the fifth; both bind the parties.

\iffull
\subsection{Currencies}\label{sec:cb}

A grassroots currencies contract~\cite{shapiro2024gc} has five clauses.  A \emph{coin} is a unit of debt naming its issuer, and an \emph{Alice-coin} is a coin issued by Alice.

\begin{enumerate}
\item \emph{Minting.}  Alice may decide to issue Alice-coins.
\item \emph{Swap.}  Alice and Bob may jointly decide to swap a coin Alice holds for a coin Bob holds.
\item \emph{Payment.}  Alice may decide to pay a Bob-coin they hold to Bob.
\item \emph{Redemption.}  Alice may decide to redeem a Bob-coin they hold against any coin Bob holds.
\item \emph{The issuer's obligation.}  Alice accepts Alice-coins, at the prices they advertise, for the goods and services they offer.
\end{enumerate}

\noindent The first four determine the contract's transactions.  A swap in which Alice gives an Alice-coin and Bob gives a Bob-coin opens a mutual credit line between them, each thereafter holding a coin of the other.  The fifth describes no act: it gives a coin its worth.

\fi

%% file: sections/03-schemas.tex
\section{Act Schemas}\label{sec:schemas}

A clause of a contract describes an \emph{act}, in the sense of a juridical act: a declaration by one or more parties intended to create a legal effect.  This section gives those acts a formal syntax.  A \emph{contract} $C$ is a finite set of act schemas over a set $\Sigma$ of predicates, and $S_C$ is its set of states.  A party's state is a multiset of atoms; a schema names an act's roles, which of them must will it, and the atoms it requires, forbids, adds and deletes at each role.

\subsection{Atoms and States}\label{sec:atoms}

Write $\Pi$ for a potentially infinite set of at least two \emph{people}, of which every contract and every run concerns a finite subset, and $D$ for a countable set of \emph{speech acts} disjoint from it: a speech act is a message, a post or an item a party signs, taken as an identity.  A speech act carries the signature of the party that signed it and so determines its signer~\cite{cardelli2020digital,golike2026digital}; the schemas take this as given.  Write $\Sigma$ for a set of \emph{predicate symbols}, \emph{predicates} for short, and $\alpha : \Sigma \rightarrow \mathbb{N}$ for their arities; a predicate of arity $n$ applied to $n$ arguments is an atom, a syntactic object with no truth value.  The arguments of an atom of a state are people and speech acts; those of an atom of a schema are roles and variables. 

\begin{definition}[Atoms and states]\label{def:states}
For any set $W$ of arguments put
\[
\Sigma[W] \;=\; \{\, e(u_1,\ldots,u_{\alpha(e)}) \;:\; e \in \Sigma,\ u_i \in W \,\},
\]
call its members the \emph{atoms over $W$}, and call $u_1,\ldots,u_{\alpha(e)}$ the \emph{arguments} and $e$ the \emph{predicate} of $e(u_1,\ldots,u_{\alpha(e)})$.  An atom over $\Pi \cup D$ is \emph{ground} , a \emph{local state} is a finite multiset of ground atoms, the initial state is $s_0 = \emptyset$, and
\[
S_C(P) \;=\; \{\, s \;:\; s \text{ a finite multiset over } \Sigma[P \cup D] \,\} \qquad (P \subset \Pi).
\]
\end{definition}

 Multiset operations are written $\uplus$, $\setminus$, $\cap$ and $\subseteq$: multiplicities are added, subtracted, minimised and compared.  A set is  a multiset in which every multiplicity is one.

A state is local, so its holder is no argument of its atoms: $\mathit{friend}(\Bob)$ in the state of Alice says that Bob is a friend of Alice, and \textcent$(u)$ in the state of Alice is a coin issued by $u$ and held by Alice; a friend set is a set of the former and a holding a multiset of the latter.

\subsection{The Syntax of a Schema}\label{sec:schemasyntax}

Fix disjoint sets $R$ of \emph{roles} and $X$ of \emph{party variables}, both disjoint from $\Pi$, and a set $Y$ of \emph{speech-act variables} disjoint from $D$.  A schema is written over the atoms over $R \cup X \cup Y$.

\iffull
A schema is a contract template: the role names are its parameters, and the identities of the parties assuming them are filled in on signature.  The mark $?$ is a question the machine puts to the person in that role~\cite{cardelli2020digital}.

\fi
For example, take one predicate $\mathit{friend}$ of arity one.  The clause that Alice and Bob may jointly decide to become friends is the schema
\[
\mathit{befriend}(\Alice?,\Bob?) : \ \ \neg\mathit{friend}(\Bob),\,+\mathit{friend}(\Bob) \ \ \ \neg\mathit{friend}(\Alice),\,+\mathit{friend}(\Alice).
\]
\iffull Read it role by role.  At Alice the atom $\mathit{friend}(\Bob)$ is forbidden, so Alice must not already have Bob as a friend, and the same atom is added, which is the effect of befriending; at Bob the same holds with Alice and Bob exchanged.  Both roles carry $?$, so the act takes the will of both.  Nothing is required and nothing deleted, and both parties change state, as every schema requires at every role.\fi

\begin{definition}[Act schema]\label{def:schema}
An \emph{act schema} is written
\[
\tau(\pi_1,\ldots,\pi_k;\,x_1,\ldots,x_m) : \quad E_1 \quad \cdots \quad E_k
\]
with $\tau$ its name, $\pi_1,\ldots,\pi_k \in R$ distinct, $k \ge 1$ its \emph{arity}, $x_1,\ldots,x_m \in X$ distinct, and each $E_i$ a finite multiset of \emph{marked atoms}, an atom over $\{\pi_1,\ldots,\pi_k,x_1,\ldots,x_m\} \cup Y$ carrying at most one of the marks $-$, $+$, $\neg$.  For each role $\pi_i$ the marks divide $E_i$ into four finite multisets of atoms,
$
\mathord{+}(i),\quad \mathord{-}(i),\quad \mathord{\neg}(i),\quad \mathord{=}(i),
$
its members marked $+$, marked $-$, marked $\neg$, and unmarked respectively; the atoms $\pi_i$ \emph{requires} are $\mathord{=}(i) \uplus \mathord{-}(i)$.  It is required that $\mathord{+}(i) \cap \mathord{-}(i) = \emptyset$ and $\mathord{+}(i) \uplus \mathord{-}(i) \neq \emptyset$.  A role may carry the mark $?$; the \emph{guarding roles} of $\tau$ are those that do.
\end{definition}

\begin{definition}[Binding]\label{def:binding}
A \emph{binding} $\beta$ for a schema $\tau$ maps its roles and party variables to $\Pi$, injectively on the roles, and $Y$ to $D$, and is extended to atoms and to multisets of them componentwise.
\end{definition}

\iffull
An act schema is a STRIPS operator~\cite{fikes1971strips} in multiagent form: what a role requires and forbids is a precondition, what it adds and deletes are an add list and a delete list, and a schema carries one of each per role, where STRIPS has one precondition and one pair of lists over a single global state.  The marks are also the four kinds of arc of a contextual net transition~\cite{montanari1995contextual}: an atom required and deleted is consumed, one required and kept is a positive context condition, one forbidden is a negative context condition, and one added is produced.  We take no result from either theory.

\fi\full{}{An act schema is a STRIPS operator~\cite{fikes1971strips} in multiagent form, and its marks are the four kinds of arc of a contextual net transition~\cite{montanari1995contextual}.
}
\subsection{The Schemas of the Social Graph}\label{sec:sg-schemas}

Three predicates.  $\mathit{friend}$, of arity one: $\mathit{friend}(\Bob)$ in the state of Alice says that Bob is a friend of Alice.  $\mathit{item}$, of arity three: $\mathit{item}(x,a,f)$ in the state of Alice is the speech act $x$, authored by $a$ and received from $f$; a speech act a party signs is authored by it and received from it.  And $\mathit{sent}$, of arity two: $\mathit{sent}(x,\Bob)$ says that its holder forwarded $x$ to Bob.  The roles are Alice and Bob, and the four clauses of Section~\ref{sec:sg} that describe acts become four schemas.  Sign carries the speech-act variable $x$; forward carries $x$ and two party variables, $a$ for its author and $f$ for the party it came from, neither of whom takes part in the act.

\begin{center}
\resizebox{\linewidth}{!}{\begin{tabular}{lll}
\hline
\emph{schema} & \emph{at Alice} & \emph{at Bob} \\
\hline
$\mathit{befriend}(\Alice?,\Bob?) :$ & $\neg\mathit{friend}(\Bob),\ +\mathit{friend}(\Bob)$ & $\neg\mathit{friend}(\Alice),\ +\mathit{friend}(\Alice)$ \\
$\mathit{unfriend}(\Alice?,\Bob) :$ & $-\mathit{friend}(\Bob)$ & $-\mathit{friend}(\Alice)$ \\
$\mathit{sign}(\Alice?;x) :$ & $+\mathit{item}(x,\Alice,\Alice)$ & \\
$\mathit{forward}(\Alice?,\Bob;x,a,f) :$ & $\mathit{item}(x,a,f),\ +\mathit{sent}(x,\Bob)$ & $\mathit{friend}(\Alice),\ +\mathit{item}(x,a,\Alice)$ \\
\hline
\end{tabular}}
\end{center}

\iffull
\noindent  ``Alice and Bob may jointly decide'' is the two guarding roles of befriend and ``if neither is the friend of the other'' its forbidden atom; ``Alice may decide'' is the single guarding role of unfriend, of sign and of forward; ``it then ends for both'' is that unfriend deletes in both roles.  ``An item they hold'' is the $\mathit{item}(x,a,f)$ forward requires and ``their friend Bob'' is the $\mathit{friend}(\Alice)$ its recipient requires; only the forwarder guards, since only the forwarder decides; and forwarding keeps the item and adds the record of the forward, as every role must change state.
\fi

\iffull
\mypara{The records} The schemas meet each of the five clauses differently.  Befriend is the only schema adding a $\mathit{friend}$ atom, that atom's argument is a role, and befriend guards in both, so a party's friends are only people who have willed the friendship, and no act of the contract can put one there against that person's will.  Unfriend is the only schema deleting a $\mathit{friend}$ atom and it deletes in both roles, so a friendship never ends for one party without the other.  A speech act a party signs is in its own state and records nothing against it; the record appears when it is first forwarded, and it is the recipient's $\mathit{item}(x,a,\Alice)$, naming the forwarder.  Sign is the only schema adding an $\mathit{item}$ atom whose author is not carried over from one it requires, and it names its own role there, so every speech act is authored by the party that signed it.  A party may forward what it received from a friend it does not personally know, so the schemas cannot stop the warranty of clause~4 being false.  Holding a $\mathit{friend}$ atom is not knowing a person.  Instead they leave the false warranty on record against the party that gave it, in the hands of the party it was given to.  Passing an item on outside the contract is no act of it, and forwarding, the one act that discloses an item, is permitted by clause~5, so no act can breach it; and a forward leaves $\mathit{item}(x,a,\Alice)$ in the recipient's hands, naming the party that passed it on, so the party bound is on record.

\fi
\iffull
\subsection{The Schemas of the Currency}\label{sec:cb-schemas}

One predicate, \textcent, of arity one: \textcent$(u)$ in the state of Alice is a coin issued by $u$ and held by Alice.  The roles are Alice and Bob.

\begin{center}
\begin{tabular}{lll}
\hline
\emph{schema} & \emph{at Alice} & \emph{at Bob} \\
\hline
$\mathit{mint}(\Alice?) :$ & $+\text{\textcent}(\Alice)$ & \\
$\mathit{swap}(\Alice?,\Bob?;u,v) :$ & $-\text{\textcent}(u),\ +\text{\textcent}(v)$ & $-\text{\textcent}(v),\ +\text{\textcent}(u)$ \\
$\mathit{pay}(\Alice?,\Bob) :$ & $-\text{\textcent}(\Bob)$ & $+\text{\textcent}(\Bob)$ \\
$\mathit{redeem}(\Alice?,\Bob;r) :$ & $-\text{\textcent}(\Bob),\ +\text{\textcent}(r)$ & $-\text{\textcent}(r),\ +\text{\textcent}(\Bob)$ \\
\hline
\end{tabular}
\end{center}

\noindent ``Alice may decide'' is the single guarding role of mint, of pay and of redeem; ``Alice and Bob may jointly decide'' is the two guarding roles of the swap; ``a Bob-coin they hold'' is that pay and redeem require of Alice a coin whose argument is Bob; ``any coin Bob holds'' is redeem's party variable $r$.  The mutual credit line of Section~\ref{sec:cb} is the swap at a binding sending $u$ to Alice and $v$ to Bob, where each gives a coin of its own issue.

\iffull
\mypara{The records} Selling at an advertised price is no act of the contract, so no act can breach clause~5 and no run shows whether the issuer honours it; but a mutual credit line leaves a coin of the issuer in another party's holding, and that coin, naming the issuer, is the record of the undertaking.  Minting records nothing against the issuer, putting the coin in its own holding; the undertaking becomes evidence only when the coin first passes to someone else, as a promissory note does when handed over.  A clause that a holding contains a coin only if its issuer willed the act that put it there is breached by pay, which adds \textcent$(\Bob)$ at Bob and guards only at Alice, and by the swap, which moves a coin between two parties by the will of the two.  An issuer is not asked when its coin moves.  A coin is therefore money rather than a promise between two people.

\fi

\fi\iffull
\subsection{The Conditions Every Act-Schema Contract Meets}\label{sec:proudhon}

Of the eight conditions of Section~\ref{sec:intro}, (1), (2) and (4) are met by every contract written as act schemas, and so is (5), in that nobody is bound who has not adopted the contract, by two properties of the language.  Both are about the protocol a contract realises, so both are stated and proved in Section~\ref{sec:compilation}, where that protocol is defined, and used here.  Conditions (3), (6), (7) and (8) are those of Section~\ref{sec:informal-conditions}.

The first property is that nothing but a party bears an effect of an act: an agent's state changes only under a schema of the contract and only in a role of it, which is Proposition~\ref{prop:consent}.

The second is that the language cannot name a person.  The roles $R$, the party variables $X$ and the speech-act variables $Y$ are disjoint from $\Pi$, and a schema is written over the atoms over them, so no schema mentions a party; and by Definition~\ref{def:states} every agent begins in the same state $s_0 = \emptyset$.  Permuting the people therefore carries the system to itself, which is Proposition~\ref{prop:anonymity}, and no party can reach a position another could not, which is Corollary~\ref{cor:privilege}.

\begin{enumerate}
\item \emph{(1) An agreement of man with man.}  Only people hold states, and by Proposition~\ref{prop:consent} an agent bears an effect of the contract only as a party and only in a role; by Corollary~\ref{cor:privilege} no party can reach a position another could not.  No text can require that an agent be a natural person, and this one does not.
\item \emph{(2) Reciprocal, (4) subject to no external authority.}  A party consents twice and both consents are its own: adopting the contract is consent to its schemas, and willing an act in a guarding role is consent to that act.  By Proposition~\ref{prop:consent} nothing binds a party outside a role of a schema, so nothing binds it that its adoption did not, and every participant in an act is a party, so no registry, no log and no timing authority takes part in one; by Corollary~\ref{cor:privilege} no party is privileged by the text.  There is nothing else in the model to be subject to.  That a party is bound only by its own promise, condition (3), is Volition, Definition~\ref{def:volition}; Proudhon's \emph{reciprocal delivery}, that each party has value for value, is a condition on the contract's content and is not carried over.
\item \emph{(5) Freely accepted and signed by all the participants.}  This is how a contract is entered into rather than what it says.  Section~\ref{sec:signing} says how: each participant signs a one-sided undertaking to keep the contract and presents it to every participant it interacts with, so every participant has signed it and, in an instance of two or more, another participant holds it (Proposition~\ref{prop:signed}); and by Proposition~\ref{prop:consent} nobody else is bound.
\end{enumerate}

\noindent The remaining conditions do not hold for every contract.  A contract with no act by which two strangers can come into relation is written as act schemas as readily as one that has such an act, and fails Introduction.  The rest of this paper takes them up.

\begin{center}
\resizebox{\linewidth}{!}{\begin{tabular}{lll}
\hline
\emph{condition} & \emph{formal counterpart} & \emph{established by} \\
\hline
(1) an agreement of man with man & every participant is a party; no privileged party & Proposition~\ref{prop:consent}, Corollary~\ref{cor:privilege} \\
(2) reciprocal & every participant is a party & Proposition~\ref{prop:consent} \\
(3) no obligation but by personal promise & volition, Definition~\ref{def:volition} & Proposition~\ref{prop:closed} \\
(4) subject to no external authority & every participant is a party & Proposition~\ref{prop:consent} \\
(5) freely accepted and signed by all & the undertaking, Section~\ref{sec:signing} & Proposition~\ref{prop:signed} \\
(6) of the nature of a contract of exchange & an unobstructed introductory act, Definitions~\ref{def:introduction}, \ref{def:unobstructed} & Proposition~\ref{prop:open} \\
(7) people free to deal with each other & an unobstructed introductory act & Proposition~\ref{prop:open} \\
(8) no external register of people & traceable provenance, Definition~\ref{def:grounded} & Lemma~\ref{lem:grounded} \\
\hline
\end{tabular}}
\end{center}
\fi

\subsection{What We Prove}\label{sec:informal-conditions}

Conditions (3), (6), (7) and (8) \full{of Section~\ref{sec:intro}}{of the abstract} become three conditions on a contract's schemas, and the protocol realising a contract meeting them is grassroots.  All three are properties of the text and are decidable, so whether a contract meets them is decided by a checker rather than left to its drafter.  This section states them and the theorem; the sections that follow provide what is needed to prove it, and Section~\ref{sec:conditions} proves it.

\iffull
A party's state names people: its friends, the issuers of the coins it holds, the authors of the items it was sent.  A state that names a person stores an identifier of them, in the words of Section~\ref{sec:intro}, and every condition below is about which people a state names.

\fi
\begin{definition}[Records]\label{def:records-atoms}
A local state \emph{records} a person $q$ if $q$ is an argument of one of its atoms.
\end{definition}

\iffull
A speech act is opaque.  $D$ is disjoint from $\Pi$, so a person named inside a signed item is not an argument of the atom carrying it, and a state holding the item does not record them.  A contract records, and its schemas can require, only the arguments of atoms, and an atom has finitely many; a chain of signatures is carried inside a speech act rather than required of one.

\fi
\mypara{Introduction} Conditions (6) and (7): any two parties must be able to come into relation by an act of the contract willed by both that leaves each recording the other, e.g.\ befriending, or a swap of coins issued by their givers.  

\begin{definition}[Introductory act]\label{def:introduction}
A schema $\tau$ of $C$ of arity two, guarded in both roles, is an \emph{introductory act} if any two people $p \neq q$ have an \emph{introduction}: a binding $\beta$ of $\tau$ into $\{p,q\}$ that sends $\pi_1$ to $p$ and $\pi_2$ to $q$, under which some atom added at $\pi_1$ names $q$ and some atom added at $\pi_2$ names $p$.
\end{definition}

\iffull
Having such an act is not enough, because it may be out of reach: it may require an atom a party cannot obtain, or forbid one that somebody else can put there and leave.  The condition is that neither happens.

\fi
\begin{definition}[Unobstructed]\label{def:unobstructed}
An introductory act $\tau$ of $C$ is \emph{unobstructed} if any two people $p \neq q$ have an introduction $\beta$ such that, at each role $\pi_i$ of $\tau$:
\begin{enumerate}
    \item every atom $\tau$ requires at $\pi_i$ is, under $\beta$, added by some schema of $C$ of arity one that is guarded in its role and requires and forbids nothing, at some binding that sends its role to $\beta(\pi_i)$;
    \item every schema of $C$ that, at some binding, adds at $\beta(\pi_i)$ an atom $\tau$ forbids at $\pi_i$ has, at that binding, the other of $p$ and $q$ in a role and adds at $\beta(\pi_i)$ an atom naming them. 
\end{enumerate}
\end{definition}

\iffull
The second clause concerns an atom somebody else can put in a party's state and leave there.  It excepts an act that puts it there while bringing the two parties together: its binding gives the other party a role, and the atom it adds names that party, so every transaction it determines changes the state of both and leaves the first recording the second, and is an interaction between them.  Befriending at the binding that swaps its two roles is one, and so is any more guarded act that specialises the introduction, e.g.\ a befriending that also takes the parents' will.  The clause requires only that the act leave one of the two recording the other, which is enough for the act to be an interaction between them, and is all the proof of Proposition~\ref{prop:open} uses of it.

\fi
\mypara{Provenance} Condition (8): a record may name a person its holder has never dealt with: a coin that changed hands, an item forwarded from a friend.  Such a record must be traceable to someone who was there: either the act that made it had that person as a party, or another party to the act already held a record naming them.

\begin{definition}[Traceable provenance]\label{def:grounded}
Write $\Phi_E$ for the atoms over $R \cup X \cup Y$ whose predicate is in $E$.  A set $E \subseteq \Sigma$ has \emph{traceable provenance} in a contract $C$ if for every schema $\tau$ of $C$ with roles $\pi_1,\ldots,\pi_k$, every $i$, every $\varphi \in \mathord{+}(i) \cap \Phi_E$ and every argument $u$ of $\varphi$ in $R \cup X$ with $u \neq \pi_i$,
\[
u \in \{\pi_1,\ldots,\pi_k\}
\quad\text{or}\quad
u \text{ is an argument of some } \psi \in (\mathord{=}(j) \uplus \mathord{-}(j)) \cap \Phi_E, \text{ some } j.
\]
A predicate has \emph{traceable provenance} in $C$ if it belongs to such a set, and an atom is said to have it when its predicate does.
\end{definition}

\iffull
A predicate whose atoms are added only with roles among their arguments has traceable provenance.  The condition is weaker: an act may give a party a record of someone it does not yet record, provided another party to that act records them already.

\fi
\iffull
\mypara{Volition} Condition (3): an act that not all its parties are asked about must run only between parties that already record each other, as do unfriending and forwarding between friends, and paying by a holder of the payee's coins.

\fi
\begin{definition}[Volition]\label{def:volition}
The \emph{role graph} of a schema is the undirected graph on its roles with an edge between $\pi_i$ and $\pi_j$ when one of the two requires an atom of traceable provenance with the other among its arguments.  A contract $C$ satisfies \emph{volition} if every schema of it of arity two or more is guarded in all its roles or has a connected role graph.
\end{definition}

\iffull
The role graph is the undirected graph formed by the identifiers the parties store, of Section~\ref{sec:intro}, taken on the roles of one schema.  For a schema of arity two it is connected exactly when one of its two roles requires such an atom naming the other.  Above that arity connectedness is weaker than requiring it of every pair, and Proposition~\ref{prop:closed} uses it: lying in one instance is an equivalence, so a path through the roles puts all the participants there.

\fi
\begin{definition}[Syntactically grassroots]\label{def:syntactically-grassroots}
A contract is \emph{syntactically grassroots} if it has an unobstructed introductory act and satisfies volition.
\end{definition}

A social contract written as act schemas is a grassroots social contract when it is syntactically grassroots, and the compilation of Section~\ref{sec:compilation} is defined for such contracts.

\begin{restatable}[Grassroots realisation]{theorem}{grassrootsthm}\label{thm:grassroots}
The protocol realising a syntactically grassroots social contract is volitionally grassroots.
\end{restatable}
The rest of the paper is dedicated to proving this theorem.

Whether a contract is syntactically grassroots is decidable, and a checker deciding it accompanies this paper, at \url{https://github.com/EShapiro2/GLP/tree/main/programs/jurix}.  \iffull It certifies both contracts of Section~\ref{sec:contracts}, with the sets of predicates of traceable provenance that Section~\ref{sec:conditions} states.  Everything the three conditions quantify over is finite once the bindings are taken up to renaming: an introductory act is an assignment of a schema's party variables to its two roles; unobstructed is the matching of one atom against those the finitely many schemas add; volition is a connectivity check on the roles.  Traceable provenance is not decided in a single pass: the largest set of predicates having it is reached by dropping any predicate that fails the condition and repeating, which ends because there are finitely many.\fi

\iffull
Section~\ref{sec:transactions} recalls what \emph{grassroots} means and Section~\ref{sec:compilation} says what the protocol realising a contract is.  Section~\ref{sec:polite} proves the theorem the proof rests on, and Section~\ref{sec:conditions} proves this one.

\fi

%% file: sections/04-transactions.tex
\section{Volition-Guarded Transactions and Grassroots Protocols}\label{sec:transactions}

This section recalls from~\cite{lewis2026volitional} what the statements below use, in brief.  Nothing in it is ours.  \full{Appendix~\ref{app:framework} states the same notions formally, together with the results used in the proofs of Section~\ref{sec:polite}, so that the paper is self-contained; the proofs of those results are in~\cite{lewis2026volitional} and are not repeated.}{The same notions are stated formally, with the results the proofs use, in the full version of this paper~\cite{shapiro2026formalising}.}
\full{}{Volition-guarded transactions follow multiagent transition systems~\cite{shapiro2021multiagent} and multiagent atomic transactions~\cite{shapiro2025atomic}.  }Throughout, $\Pi$ is a potentially infinite set of \emph{agents}, an agent being a person operating a machine, and $P \subset \Pi$ ranges over its finite nonempty subsets.  For a set $\calS$, $\calS^P$ is the set of total functions from $P$ to $\calS$, and $c_p$ is the value of $c$ at $p$.

A \emph{local-states function} $S$ maps every $P \subset \Pi$ to a set $S(P)$ of local machine states containing the initial state $s_0$, with $P \subset P'$ implying $S(P) \subset S(P')$.  A \emph{machine transaction} over \emph{participants} $Q$ is a pair $d \rightarrow d' \in (\calS^Q)^2$ with $d \neq d'$; it is \emph{unary} if $|Q| = 1$.  A \emph{volition-guarded transaction} is a pair $(t,Q')$ of a machine transaction $t$ over $Q$ with its \emph{guards} $Q' \subseteq Q$: the people in $Q'$ must be willing for it to be taken.  A \emph{transaction equivalence} $\sim$ relates transactions with the same participants and has countably many classes; a person wills a class, and taking one member of it fulfils that will.  For a set $R$ of volition-guarded transactions over $S$ and $P \subset \Pi$, we write $R(P)$ for those of them whose transactions are over $S(P')$ for some $P' \subseteq P$.

An \emph{agent state} is a pair of a \emph{volitional state}, a set of classes the person is willing their machine to take part in, and a machine state.  A set of volition-guarded transactions over a local-states function induces, for each $P$, a transition system whose configurations are agent configurations over $P$ and whose transitions are \emph{volition changes}, by which a person by themselves alters their volitional state, and \emph{volitional machine transactions}, each induced by a volition-guarded transaction whose guards all hold its class.  A transaction is \emph{machine-enabled} in a configuration whose machine states restrict on the participants to its source, and \emph{enabled} if in addition every guard holds its class.  The family of these systems, one per $P$, is the \emph{protocol} $\calF$ over the set, $\calF(P)$ being its member at $P$; runs are \emph{safe} when consecutive states form transitions, \emph{live} when no class stays enabled forever without a member being taken, and \emph{correct} when both.

\full{\begin{definition}[Interaction, Interactive Run, First Interaction~\cite{lewis2026volitional}]\label{def:first-interaction}
Let $Q\subset\Pi$.  For a configuration $c$ over $Q$ and $Q'\subseteq Q$, write $c|_{Q'}$ for the restriction of $c$ to $Q'$.  A transition $e\to e'$ of $\calF(Q)$ is an \emph{interaction between} $x\ne y\in Q$ if $e|_{\{x\}} \ne e'|_{\{x\}}$ and $e|_{\{y\}} \ne e'|_{\{y\}}$, and $e|_{\{x\}} \to e'|_{\{x\}}$ is not a transition of $\calF(\{x\})$ or $e|_{\{y\}} \to e'|_{\{y\}}$ is not a transition of $\calF(\{y\})$.  It is an \emph{interaction between} disjoint nonempty $P, P'\subseteq Q$ if it is an interaction between some $x\in P$ and some $y\in P'$.  A run $\hat r$ of $\calF(Q)$ is \emph{interactive between $P$ and $P'$} if some transition of $\hat r$ is an interaction between them; the \emph{first interaction} of $P$ and $P'$ in $\hat r$ is the first such transition.
\end{definition}

\begin{definition}[Interaction Graph, Instance, Coalescence~\cite{lewis2026volitional}]\label{def:coalescence}
Let $\calF$ be a protocol, $Q\subset\Pi$, and $r$ a prefix of a run of $\calF(Q)$.  The \emph{interaction graph} of $r$ has vertices $Q$ and an edge between $a\ne b\in Q$ whenever some transition of $r$ is an interaction between $a$ and $b$.  A nonempty $P\subseteq Q$ is an \emph{instance} in $r$ if it is a connected component of the interaction graph of $r$.  Disjoint nonempty $P, P'\subseteq Q$ \emph{coalesce} in a run $\hat r$ of $\calF(Q)$ if they are instances in some prefix of $\hat r$ and lie in one component of the interaction graph of a longer prefix of $\hat r$.
\end{definition}}{A transition of $\calF(Q)$ is an \emph{interaction} between two agents if it changes the state of each and its restriction to one of them is no transition of that agent's own system; it is an interaction between two disjoint groups if it is one between a member of each.  In the \emph{interaction graph} of a run, two agents are joined whenever some transition of the run is an interaction between them; an \emph{instance} in the run is a connected component of that graph, and two instances \emph{coalesce} when they come to lie in one component.}

A protocol is \emph{oblivious} if any interleaving of correct runs of two disjoint groups is a correct run of the two together, \emph{interactive} if two groups that have not yet interacted can always still do so, and \emph{grassroots} if it is both.  It is \emph{volitionally grassroots} if it is grassroots and the first interaction of two groups is induced by a transaction guarded by a member of each, so groups become connected only by mutual consent.  In a grassroots protocol two instances can always coalesce.

%% file: sections/05-compilation.tex
\section{Compiling Schemas into Transactions}\label{sec:compilation}

Act schemas have no operational semantics of their own: a schema of a syntactically grassroots contract means the volition-guarded transactions it compiles to, and those have the semantics of~\cite{lewis2026volitional}.  The compilation is defined for syntactically grassroots contracts, so a contract is checked before it is compiled.  $S_C$ is a local-states function\full{ in the sense of Appendix~\ref{app:framework}}{}: $s_0 \in S_C(P)$ for every $P$, and $P \subset P'$ implies $S_C(P) \subset S_C(P')$, strictly, provided some predicate has positive arity, as one does in every syntactically grassroots contract, its introductory act adding at $\pi_1$ an atom naming $\pi_2$: for $q \in P' \setminus P$ and such a predicate, an atom with $q$ among its arguments gives a state in $S_C(P')$ and not in $S_C(P)$.

\subsection{The Compilation}\label{sec:compile}

Compilation maps an act schema of a syntactically grassroots contract to a volition-guarded multiagent atomic transaction~\cite{lewis2026volitional}.

\begin{definition}[Compilation]\label{def:compile}\label{def:instance}
Let $\tau$ be a schema with roles $\pi_1,\ldots,\pi_k$, $\beta$ a binding for it, and $p_i := \beta(\pi_i)$.  A machine transaction $d \rightarrow d'$ over $\{p_1,\ldots,p_k\}$ is \emph{determined by $\tau$ at $\beta$} if for every $i$
\[
\begin{gathered}
\beta(\mathord{=}(i) \uplus \mathord{-}(i)) \subseteq d_{p_i}, \qquad
\beta(\mathord{\neg}(i)) \cap d_{p_i} = \emptyset, \qquad
\beta(\mathord{+}(i)) \cap \beta(\mathord{-}(i)) = \emptyset, \\
d'_{p_i} = \bigl(d_{p_i} \setminus \beta(\mathord{-}(i))\bigr) \uplus \beta(\mathord{+}(i)),
\end{gathered}
\]
and the volition-guarded transaction it determines is the pair of it with $\{\beta(\pi) : \pi \text{ a guarding role of } \tau\}$.\iffull{}  In the notation of~\cite{lewis2026volitional}, $\tau$ \emph{compiles to}
\[
\begin{array}{ll}
c'_{\pi_i} := (c_{\pi_i} \setminus \mathord{-}(i)) \uplus \mathord{+}(i) & (1 \le i \le k),\\[2pt]
\text{provided } \mathord{=}(i) \uplus \mathord{-}(i) \subseteq c_{\pi_i} \text{ and } \mathord{\neg}(i) \cap c_{\pi_i} = \emptyset & (1 \le i \le k),\\[2pt]
\text{guarded by the guarding roles of } \tau, &
\end{array}
\]
which at a binding $\beta$ denotes the volition-guarded transactions determined by $\tau$ at $\beta$.\fi
\end{definition}

A contract compiles schema by schema.  A schema with no guarding role compiles to transactions with an empty guard, which are enabled whenever they are machine-enabled and are then eventually taken: no person wills such an act.

\begin{definition}[Realisation of a contract]\label{def:realisation}
Let $C$ be a syntactically grassroots contract.  Its \emph{realisation} $R_C$ is the set of volition-guarded transactions determined by the schemas of $C$ at all bindings for them, and $\sim$ is the equivalence on their underlying machine transactions generated by relating two when one schema at one binding determines both.  The \emph{protocol realising $C$} is the volitional transactions-based protocol over $R_C$, $S_C$ and $\sim$.
\end{definition}

\iffull
A set of volition-guarded transactions over $S_C$ with an equivalence $\sim$ is \emph{well formed} if distinct machine transactions underlying it have disjoint $P$-closures for every $P \subset \Pi$, and $\sim$ is a transaction equivalence with countably many classes over any $P$, which is what Definitions~\ref{def:equivalence} and~\ref{def:vmts} require.

\begin{lemma}[The compilation is well defined]\label{lem:translation}
$R_C$ with $\sim$ is well formed.
\end{lemma}

\fi
\full{}{$R_C$ is a set of volition-guarded transactions over $S_C$ whose distinct machine transactions have disjoint $P$-closures, and $\sim$ is a transaction equivalence with finitely many classes over any $P$.}
\iffull
\begin{proof}
Let $d \rightarrow d'$ be determined by $\tau$ at $\beta$, with participants $Q$ the people in its roles, and let $P'$ be $Q$ together with the people $\beta$ assigns to the party variables of $\tau$ and the arguments of the atoms of $d$.  Every atom of every $d'_{p}$ is an atom of $d_p$ or the image under $\beta$ of an atom added in a role, so its arguments lie in $P'$, as do those of every atom of $d$, and $d \rightarrow d'$ is a machine transaction over $Q$ and $S_C(P')$.

By Definition~\ref{def:compile} the state of every participant changes, so the participants of a machine transaction are exactly the agents that are not stationary in any transition of its $P$-closure.  A transition lying in the closures of two of them therefore forces the two to have the same participants and the same restriction to them, hence to be equal, and the closures of distinct machine transactions are disjoint, which is what Definition~\ref{def:vmts} requires of $R_C$.

Two machine transactions determined by one schema at one binding have the participants of that binding, so the relation generating $\sim$ holds only between transactions with the same participants, and so does $\sim$.  A class over $P$ is a union of sets of transactions each determined by a schema at a binding sending its roles and party variables into $P$ and its speech-act variables into $D$; a contract has finitely many schemas, each has finitely many such bindings of its roles and party variables, and $D$ is countable, so the classes over $P$ are countably many, as Definition~\ref{def:equivalence} requires.
\end{proof}
\fi

\iffull
\subsection{Two Examples}\label{sec:compiled-examples}

With Alice and Bob standing for the parties in those roles, befriend (Section~\ref{sec:sg-schemas}) compiles to
\begin{align*}
& c'_{\Alice} := c_{\Alice} \uplus \{\mathit{friend}(\Bob)\}, \qquad c'_{\Bob} := c_{\Bob} \uplus \{\mathit{friend}(\Alice)\},\\
& \text{provided } \mathit{friend}(\Bob) \notin c_{\Alice} \text{ and } \mathit{friend}(\Alice) \notin c_{\Bob}, \qquad \text{guarded by } \{\Alice,\Bob\},
\end{align*}
and denotes, for every $d$ with $\mathit{friend}(\Bob) \notin d_{\Alice}$ and $\mathit{friend}(\Alice) \notin d_{\Bob}$, the volition-guarded transaction $(d \rightarrow d',\{\Alice,\Bob\})$ with $d'_{\Alice} = d_{\Alice} \uplus \{\mathit{friend}(\Bob)\}$ and $d'_{\Bob} = d_{\Bob} \uplus \{\mathit{friend}(\Alice)\}$.  There is one for every admissible $d$, since the two may befriend whatever other friends each has, and $\sim$ relates them all: they are one act of befriending, available at different points of a run.  \iffull The swap (Section~\ref{sec:cb-schemas}) compiles to
\begin{align*}
& c'_{\Alice} := (c_{\Alice} \setminus \{\text{\textcent}(u)\}) \uplus \{\text{\textcent}(v)\}, \qquad c'_{\Bob} := (c_{\Bob} \setminus \{\text{\textcent}(v)\}) \uplus \{\text{\textcent}(u)\},\\
& \text{provided } \text{\textcent}(u) \in c_{\Alice} \text{ and } \text{\textcent}(v) \in c_{\Bob}, \qquad \text{guarded by } \{\Alice,\Bob\}.
\end{align*}
The bindings with $u = v$ determine nothing, by Definition~\ref{def:compile}: a swap of a coin for a coin of the same issue changes no state.  The swap at $u = \Alice$ and $v = \Bob$ is the mutual credit line.\fi

\fi
\subsection{Every Participant Is a Party}\label{sec:party}

\iffull
Throughout, $C$ is a syntactically grassroots contract, $\calF$ the protocol realising it and $P \subset \Pi$.  No act reaches anyone who is not a party.

\begin{proposition}[Every participant is a party]\label{prop:consent}
Every agent whose machine state a transition of $\calF(P)$ changes is assigned to a role of a schema of $C$ by some binding.
\end{proposition}

\iffull
\begin{proof}
A volition change leaves every machine state unchanged (Definition~\ref{def:vmat}), so a transition changing one is a volitional machine transaction induced by some $(t,Q')\in R_C(P)$, and by that definition it leaves the machine state of every agent outside the participants of $t$ unchanged.  By Definition~\ref{def:realisation} $t$ is determined by a schema $\tau$ of $C$ at a binding $\beta$, and by Definition~\ref{def:compile} the participants of $t$ are the people $\beta$ assigns to the roles of $\tau$.
\end{proof}
\fi

An agent therefore bears an effect of the contract only under a schema of it, and only in a role.  The second property\full{ of Section~\ref{sec:proudhon}}{} is that the schemas single out no agent.

A permutation $\sigma$ of $\Pi$ extends to atoms, by permuting the arguments in $\Pi$ and fixing those in $D$, and to states, configurations, transactions and runs componentwise.

\begin{proposition}[Anonymity]\label{prop:anonymity}
For every permutation $\sigma$ of $\Pi$, $\sigma(R_C) = R_C$, $\sigma(c_0(P)) = c_0(\sigma(P))$, and $\sigma$ carries the runs of $\calF(P)$ to the runs of $\calF(\sigma(P))$.
\end{proposition}

\iffull
\begin{proof}
$\sigma$ maps $\Sigma[P \cup D]$ bijectively onto $\Sigma[\sigma(P) \cup D]$, hence $S_C(P)$ onto $S_C(\sigma(P))$, and fixes $s_0 = \emptyset$; so $\sigma(c_0(P)) = c_0(\sigma(P))$.  Let $t$ be determined by a schema $\tau$ at a binding $\beta$.  Then $\sigma \circ \beta$ is again a binding for $\tau$, being injective on the roles and leaving $Y$ mapped into $D$.  Each clause of Definition~\ref{def:compile} is an inclusion, a disjointness or an identity between images under $\beta$ and the local states of $d$ and $d'$, and $\sigma$ is a bijection on atoms, so applying it to both sides of each shows $\sigma(t)$ determined by $\tau$ at $\sigma \circ \beta$.  The guarding roles of $\tau$ do not depend on the binding, so $\sigma$ carries the volition-guarded transaction at $\beta$ to the one at $\sigma \circ \beta$; by Definition~\ref{def:realisation} $\sigma(R_C) \subseteq R_C$, and $\sigma^{-1}$ gives the reverse inclusion.  Two members of $R_C$ are $\sim$-related exactly when one schema at one binding determines both, a condition $\sigma$ preserves, and $R_C(P)$ is carried to $R_C(\sigma(P))$ because $\sigma$ carries $S_C(P')$ to $S_C(\sigma(P'))$ for every $P' \subseteq P$.  The volitional construction and the $P$-closure are defined from $R_C(P)$, $\sim$ and $c_0(P)$ only, so $\sigma$ carries $\calF(P)$ to $\calF(\sigma(P))$ and runs to runs.
\end{proof}
\fi

\iffull
\begin{corollary}[No privileged party]\label{cor:privilege}
If a configuration is reachable in a run over $P$, so is its image under every permutation of $P$.
\end{corollary}

\fi
\iffull
\begin{proof}
Extend the permutation of $P$ by the identity on $\Pi \setminus P$ and apply Proposition~\ref{prop:anonymity}, which carries the run to a run over $P$ ending in the image configuration.
\end{proof}
\fi
\fi
\full{}{No act reaches anyone who is not a party: every agent whose state a transition of the protocol realising $C$ changes is the person some binding assigns to a role of some schema of $C$.  And the schemas single out no agent: a permutation of the people carries the protocol to itself, so a configuration reachable in a run over $P$ has every image under a permutation of $P$ reachable too.}

Conditions (1), (2), (4) and (5) \full{of Section~\ref{sec:intro}}{of the abstract} follow from these two properties\full{, as Section~\ref{sec:proudhon} sets out}{}.  That the protocol realising the contract is grassroots does not follow from them.

%% file: sections/06-polite.tex
\section{Polite Sets of Volition-Guarded Transactions}\label{sec:polite}

Here we prove that the following condition on a set of volition-guarded transactions makes the protocol over it volitionally grassroots: (1) any two agents can always come to interact, and (2) no act joins two instances unless all its participants will it.  Section~\ref{sec:conditions} states conditions on the acts of a contract sufficient for its compiled volition-guarded transactions to satisfy it, and both contracts of Section~\ref{sec:contracts} satisfy it.

\iffull
The proofs use two graphs: the interaction graph of Section~\ref{sec:transactions}, undirected, whose connected components are the instances; and the directed graph of who holds a record of whom, defined next.

\fi
Throughout this section, $S$ is a local-states function, $R$ a set of volition-guarded transactions over $S$ with equivalence $\sim$, and $\calF$ the protocol over $R$, $S$ and $\sim$.  \full{The proofs use four results of~\cite{lewis2026volitional} recalled in Appendix~\ref{app:framework}: obliviousness rests on Lemmas~\ref{lem:interleaving-safety} and~\ref{lem:volitional-containment} through Proposition~\ref{prop:volitional-oblivious}, and interactivity on Lemma~\ref{lem:extension}; openness yields interactivity and closure yields consent.}{} 

\begin{definition}[Records]\label{def:records}
A machine state $s$ \emph{records} an agent $q\in\Pi$ if $s\notin S(P)$ for every $P\subset\Pi$ with $q\notin P$.
\end{definition}

\begin{definition}[Recording graph]\label{def:recording-graph}
The \emph{recording graph} of a machine configuration $d$ is the directed graph on $\Pi$ with an edge $p\rightarrow q$ whenever $d_p$ records $q$.  The recording graph of an agent configuration is that of its machine states, and the recording graph of a run is the union of the recording graphs of its configurations.
\end{definition}

\begin{definition}[Mutual recording]\label{def:mutual}
A volition-guarded transaction $(d\rightarrow d',\{p,q\})$ over participants $\{p,q\}$, $p\ne q\in\Pi$, is a \emph{mutual recording of $p$ and $q$} if it changes the state of each and the recording graph of $d'$ has the edges $p\rightarrow q$ and $q\rightarrow p$. 
\end{definition}

\iffull
A state records $q$ when it could not have arisen among agents that leave $q$ out: every set of agents whose states include it contains $q$.  The recording graph of the initial configuration is empty, since $s_0\in S(P)$ for every $P$.  In the social graph of Section~\ref{sec:sg-schemas} the edges out of a party are its friends\full{; in the currency of Section~\ref{sec:cb-schemas} they are the issuers of the coins it holds}{}.  A mutual recording is an act of two agents, willed by both, after which there is an edge each way between them: befriending is one\full{, and a swap of coins the two have issued is another}{}.  An edge may point at an agent its source has never transacted with, as when a message signed by one friend is forwarded by another or a coin changes hands, and nothing below forbids it.  The recording graph of a configuration is not monotone along a run, as an act may delete the atoms that put an edge there; the recording graph of a run and the interaction graph are monotone by construction.

Let $Q\subset\Pi$.

\begin{lemma}[A Mutual Recording is an Interaction]\label{lem:recording-interaction}
Every transition of $\calF(Q)$ induced by a mutual recording of $p$ and $q$ in $R(Q)$ is an interaction between $p$ and $q$.
\end{lemma}

By Definition~\ref{def:first-interaction} it is then an interaction between any two disjoint sets of agents one of which holds $p$ and the other $q$.
\fi

\iffull
\begin{proof}
The machine states of $p$ and of $q$ both change, by Definition~\ref{def:mutual}.  Since $d'_p$ records $q$, $d'_p\notin S(P)$ for every $P$ with $q\notin P$, and in particular $d'_p\notin S(\{p\})$; the machine states of every configuration of $\calF(\{p\})$ lie in $S(\{p\})$ by Definition~\ref{def:protocol-transactions}, so the restriction of the transition to $p$ is no transition of $\calF(\{p\})$.  By Definition~\ref{def:first-interaction} the transition is an interaction between $p$ and $q$, hence between any two disjoint sets of agents one of which holds $p$ and the other $q$.
\end{proof}
\fi

\begin{definition}[Open, Closed, Polite]\label{def:polite}
The set $R$ is:
\begin{enumerate}
    \item \emph{open} if for any two agents $p\ne q$ in any $Q\subset\Pi$, every finite safe run of $\calF(Q)$ with no interaction between $p$ and $q$ has a finite safe extension, by volition changes and unary transitions, at whose end some transaction of $R(Q)$ whose induced transitions are interactions between $p$ and $q$ is machine-enabled;
    \item \emph{closed} if for every $Q\subset\Pi$, every transaction of $R$ enabled at the end of a finite safe run of $\calF(Q)$ is guarded by all its participants or has them all in one instance in that run;
    \item \emph{polite} if it is open and closed.
\end{enumerate}
\end{definition}

\iffull
Openness says that any two agents can always be brought, by zero or more acts each takes independently, to a configuration at which an act coupling the two can be taken if all its guards are willing.  It does not require that no one else take part in the act: a mutual recording is one, and so is any act that leaves one of them holding a record of the other, whoever else takes part.  Of two agents that have already taken part in one act nothing is required: they have interacted, and being interactive requires no more.  Closure says that an act not willed by all its participants is enabled only where they already lie in one instance.

\fi

\begin{theorem}[Polite]\label{thm:polite}
The protocol over a polite set of volition-guarded transactions is volitionally grassroots.
\end{theorem}

\iffull
\begin{proof}
Let $R$ be polite.  One fact is used throughout: an agent whose state a volitional machine transaction changes is a participant of the transaction inducing it, as only a participant changes machine state and only a participant holds its class in its volitional state (Definitions~\ref{def:agent-state} and~\ref{def:vmat}).

\mypara{Oblivious} Let $P, P'\subset\Pi$ be disjoint and nonempty, $r$ a correct run of $\calF(P)$, $r'$ a correct run of $\calF(P')$, and $e = e_0,e_1,\ldots$ an interleaving of $r$ and $r'$, with index sequences $(i_k)$ and $(j_k)$.  By Lemma~\ref{lem:interleaving-safety}, every finite prefix of $e$ is a finite safe run of $\calF(P\cup P')$.  Every transition of $e$ is a $P$-step or a $P'$-step and so alters the state of agents of one group only, so no transition of $e$ is an interaction between an agent of $P$ and an agent of $P'$, and no instance in a prefix of $e$ holds agents of both.

We verify the hypothesis of Proposition~\ref{prop:volitional-oblivious}.  Let $[t]$ be a class of $T_{R(P\cup P')}/{\sim}$ whose transactions have participants $Q$ meeting both $P$ and $P'$---the same $Q$ for every member of the class, by Definition~\ref{def:equivalence}---let $(t'',Q'')\in R$ with $t''\in[t]$, and take $a\in Q\cap P$.  Suppose $(t'',Q'')$ is enabled at some $e_k$.  Its participants $Q$ meet both groups and no instance in the prefix of $e$ ending at $e_k$ holds agents of both, so they do not all lie in one instance in that prefix; by closure $Q''=Q$, and hence $a\in Q''$.  Since $Q$ meets $P'$, $t''\notin R(P)$, so $[t]\notin T_{R(P)}/{\sim}$; and by Lemma~\ref{lem:volitional-containment} applied to $r$, $e_k{^v_a} = (c_{i_k})^v_a\subseteq T_{R(P)}/{\sim}$.  Hence $[t]\notin e_k{^v_a}$ and the guard on $a$ fails, so $(t'',Q'')$ is not enabled at $e_k$---a contradiction.

So no class whose transactions have participants in both groups is ever enabled in $e$, and by Proposition~\ref{prop:volitional-oblivious}, $\calF$ is oblivious.

\mypara{Interactive} Let $Q\subset\Pi$, let $P, P'\subseteq Q$ be disjoint and nonempty, let $\hat r$ be a prefix of a correct run of $\calF(Q)$ containing no interaction between $P$ and $P'$, and take $p\in P$ and $q\in P'$.  An interaction between $p$ and $q$ is one between $P$ and $P'$ by Definition~\ref{def:first-interaction}, so $\hat r$ contains none, and by openness some finite safe extension of $\hat r$ by volition changes and unary transitions ends where some $(t,Q')\in R(Q)$ whose induced transitions are interactions between $p$ and $q$ is machine-enabled.  A volition change and a unary transition each alter one agent only, and are therefore no interaction between $P$ and $P'$.

Extend further by a volition change of each agent of $Q'$, where needed, so that $[t]$ lies in the volitional state of each; these leave every machine state unchanged, so $(t,Q')$ is now enabled.  Extend by the volitional machine transaction it induces, an interaction between $p$ and $q$ and hence between $P$ and $P'$.  The result is a finite safe run of $\calF(Q)$ having $\hat r$ as a prefix and containing an interaction between $P$ and $P'$, and by Lemma~\ref{lem:extension} it is a prefix of a correct run of $\calF(Q)$.  Hence $\calF$ is interactive, and with obliviousness, grassroots.

\mypara{Volitionally grassroots} Let $P, P'\subset\Pi$ be disjoint and nonempty, let $\hat r$ be a safe run of $\calF(P\cup P')$ interactive between $P$ and $P'$, and let $e\rightarrow e'$ be its first interaction.  It changes the state of an agent of $P$ and of an agent of $P'$, and a volition change alters one agent only, so it is a volitional machine transaction induced by some $(t,Q')\in R(P\cup P')$ whose participants $Q$ meet both groups; take $a\in Q\cap P$ and $b\in Q\cap P'$.

The participants of $t$ do not all lie in one instance in the prefix $\rho$ of $\hat r$ ending at $e$: were they, $a$ and $b$ would lie in one connected component of the interaction graph of $\rho$; every agent of $\rho$ is in $P\cup P'$, so a path from $a$ to $b$ in that component has an edge between some $x\in P$ and some $y\in P'$, and the transition of $\rho$ that puts that edge there is an interaction between $x$ and $y$, hence between $P$ and $P'$, and it precedes $e\rightarrow e'$---contradicting that $e\rightarrow e'$ is the first.  As $(t,Q')$ is enabled at $e$, by closure $Q'=Q$, which meets both $P$ and $P'$.  Hence $\calF$ is volitionally grassroots.
\end{proof}
\fi

\iffull
Politeness is a condition on the transactions of a protocol.  Being grassroots is a property of its runs.  The theorem states that the first is sufficient for the second.  Section~\ref{sec:conditions} takes the step before it, from the acts of a contract to the transactions realising them.

\fi

%% file: sections/07-conditions.tex
\section{Conditions Sufficient for Openness and Closure}\label{sec:conditions}

We are now ready to prove Theorem~\ref{thm:grassroots}, restated.  Openness and closure are conditions on runs; the three conditions are on the schemas, and the proof has two halves: openness follows from the unobstructed introductory act, and closure from volition.  Throughout this section $C$ is a syntactically grassroots contract and $\calF$ the protocol realising it.

\grassrootsthm*

\begin{lemma}[Records]\label{lem:records}
On local states, Definitions~\ref{def:records} and~\ref{def:records-atoms} agree.
\end{lemma}

\iffull
\begin{proof}
Suppose $q$ is an argument of an atom of $s$, as Definition~\ref{def:records-atoms} requires, and let $P \subset \Pi$ with $q \notin P$.  That atom has an argument outside $P$, so $s \notin S_C(P)$ by Definition~\ref{def:states}.  As $P$ was arbitrary, $s$ records $q$.  Suppose $q$ is an argument of no atom of $s$, and let $P$ be the people among the arguments of the atoms of $s$ together with some person other than $q$, which exists as $\Pi$ has at least two people.  Then $P$ is nonempty, $q \notin P$ and $s \in S_C(P)$, so $s$ does not record $q$.
\end{proof}
\fi

\iffull
The recording graph of a configuration of the protocol realising $C$ therefore has an edge $p \rightarrow q$ exactly when some atom of the state of $p$ has $q$ among its arguments.  Its edges are classified by the predicate of the atom.

\begin{definition}[$E$-recording graph]\label{def:erecording}
Let $E$ be a set of predicates.  The \emph{$E$-recording graph} of a configuration is the spanning subgraph of its recording graph whose edges $p \rightarrow q$ are those for which the state of $p$ contains an atom with predicate in $E$ having $q$ among its arguments; the $E$-recording graph of a run is the union of those of its configurations.
\end{definition}

\fi
\begin{proposition}[Openness]\label{prop:open}
$R_C$ is open.
\end{proposition}

\iffull
\begin{proof}
Let $\tau$ be an unobstructed introductory act of $C$, let $p \neq q$ in $Q$, let $r$ be a finite safe run of the protocol over $Q$ containing no interaction between $p$ and $q$, and let $\beta$ be an introduction of $p$ and $q$ meeting Definition~\ref{def:unobstructed}.  Extend $r$, for each role $\pi_i$ of $\tau$ and each $\varphi \in \mathord{=}(i) \uplus \mathord{-}(i)$, by a volition change of $\beta(\pi_i)$ and a unary transition adding $\beta(\varphi)$ there, induced by a schema of arity one at a binding sending its role to $\beta(\pi_i)$ and its party variables into $Q$, which that agent takes by itself and which is machine-enabled, the schema requiring and forbidding nothing.  Such a schema deletes nothing, so the extension is finite and safe and every atom $\tau$ requires at $\beta$ is present at its last configuration.  A schema of arity one has neither of $p$ and $q$ in a role other than the one sent to $\beta(\pi_i)$, so by clause~2 of Definition~\ref{def:unobstructed} no atom it adds at $\beta(\pi_i)$ is one $\tau$ forbids at $\pi_i$: the extension adds no forbidden atom, and no atom $\tau$ requires at $\beta$ is one it forbids.

Let $\varphi \in \mathord{\neg}(i)$ and suppose the state of $\beta(\pi_i)$ there contains $\beta(\varphi)$.  The initial state is empty and the extension added no forbidden atom, so some transition of $r$ added it; that transition is induced by a transaction of $R_C$ determined by a schema $\tau'$ at a binding $\beta'$ with $\beta(\varphi) \in \beta'(\mathord{+}(j))$ for a role sent to $\beta(\pi_i)$.  By Definition~\ref{def:unobstructed} $\beta'$ sends a role of $\tau'$ to the other of $p$ and $q$, and $\beta'(\mathord{+}(j))$ names that party.  Both are then participants, so by Definition~\ref{def:instance} the state of each changes, and the state of $\beta(\pi_i)$ afterwards records the other, so it lies in no $S(P)$ omitting them and the transition restricted to $\beta(\pi_i)$ is no transition of that agent's own transition system: the transition is an interaction between $p$ and $q$, which $r$ does not contain.  So no forbidden atom is present, $\tau$ at $\beta$ is machine-enabled at the last configuration of the extension, and every transaction it determines there leaves $p$ holding an atom naming $q$ and $q$ one naming $p$, changing both states, and is guarded by both; by Lemma~\ref{lem:records} each records the other, so it is a mutual recording of $p$ and $q$ in $R_C(Q)$, $\beta$ mapping into $\{p,q\} \subseteq Q$, and by Lemma~\ref{lem:recording-interaction} the transitions it induces are interactions between $p$ and $q$.
\end{proof}
\fi

\iffull
Closure requires that the participants of an unguarded act already lie in one instance.  Where an atom can be added only by an act of the two people it relates, that follows directly.  An edge may also come from a transfer: a coin passes from hand to hand, and the person it points at is reached through the parties that carried it rather than by a single act.  The lemma below covers both: no act ever leaves an edge whose ends lie in different instances.

Let $E$ have traceable provenance in $C$, $Q \subset \Pi$, and $r$ a finite safe run of $\calF(Q)$.

\begin{lemma}[No edge crosses instances]\label{lem:grounded}
Every edge of the $E$-recording graph of $r$ joins two agents that lie in one instance in $r$.
\end{lemma}

\fi
\iffull
\begin{proof}
The $E$-recording graph of $r$ is the union of those of its configurations, and every configuration of $r$ ends a prefix of $r$; the interaction graph of a prefix of $r$ is a subgraph of that of $r$, so two agents in one instance in a prefix are in one instance in $r$.  It suffices to treat an edge $p \rightarrow q$ of the $E$-recording graph of the last configuration.  The claim is trivial for $q = p$, and we argue by induction on the length of $r$.  The initial state is empty, so a run of one configuration holds it vacuously.  Let $r$ end with the transition $e \rightarrow c$, and let the state of $p$ at $c$ contain an atom with predicate in $E$ and $q \neq p$ among its arguments.  If the state of $p$ at $e$ contains such an atom, the induction hypothesis applied to the prefix of $r$ ending at $e$ gives the claim.  Otherwise $e \rightarrow c$ added it, so it is a volitional machine transaction induced by a transaction determined by a schema $\tau$ at a binding $\beta$, and the atom is $\beta(\varphi)$ for an atom $\varphi$ added at a role $i$, with predicate in $E$, where $p = \beta(i)$, and $q = \beta(s)$ for an argument $s$ of $\varphi$; as $\beta(s) \neq \beta(i)$, $s \neq i$.

By Definition~\ref{def:instance} the state of every participant changes, and by Lemma~\ref{lem:records} the state of $p$ at $c$ records $q$, so it lies in no $S_C(P)$ with $q \notin P$ and the restriction of the transition to $p$ is no transition of $\calF(\{p\})$.  By Definition~\ref{def:first-interaction} the transition is therefore an interaction between $p$ and every other participant.

By Definition~\ref{def:grounded} one of two cases holds.  If $s$ is a role of $\tau$ then $q$ is a participant, so $p$ and $q$ are joined by an edge of the interaction graph of $r$.  Otherwise $s$ is an argument of an atom $\psi$ required at some role $j$, with predicate in $E$; put $p' := \beta(j)$.  By Definition~\ref{def:instance} $\beta(\psi)$ is an atom of the state of $p'$ at $e$, its predicate is in $E$ and it has $q$ among its arguments, so by the induction hypothesis $p'$ and $q$ lie in one instance in the prefix ending at $e$, hence in $r$.  If $p' \neq p$ then $p'$ is a participant, so $p$ and $p'$ are joined by an edge of the interaction graph of $r$.  Either way $p$ and $q$ lie in one instance in $r$.
\end{proof}
\fi

\begin{proposition}[Closure]\label{prop:closed}
$R_C$ is closed.
\end{proposition}

\iffull
\begin{proof}
Let $(t,Q')$ of $R_C$ be enabled at the last configuration of a finite safe run $r$ of the protocol, determined by a schema $\tau$ at a binding $\beta$.  If $\tau$ is guarded in all its roles then $(t,Q')$ is guarded by all its participants.  If $\tau$ has arity one, its single participant is a connected component of the interaction graph of $r$ by itself or lies in one, so all participants of $t$ lie in one instance in $r$.  Otherwise take an edge of the role graph of $\tau$, joining $\pi_i$ to $\pi_j$, let $\pi_i$ be the one requiring an atom of traceable provenance with $\pi_j$ among its arguments, and put $p := \beta(\pi_i)$ and $q := \beta(\pi_j)$.  As $(t,Q')$ is machine-enabled at the last configuration of $r$, the state of $p$ there contains the image under $\beta$ of that atom, so the $E$-recording graph of $r$ has the edge $p \rightarrow q$ for a set $E$ of traceable provenance containing that atom's predicate, and by Lemma~\ref{lem:grounded} $p$ and $q$ lie in one instance in $r$.  Instances are the connected components of the interaction graph of $r$, so lying in one instance is an equivalence; the role graph of $\tau$ is connected, and $\beta$ carries a path in it to a chain of agents consecutive ones of which lie in one instance, so all participants of $t$ lie in one instance in $r$.
\end{proof}
\fi

\begin{proof}[Proof of Theorem~\ref{thm:grassroots}]
Let $C$ be syntactically grassroots.  By Proposition~\ref{prop:open} $R_C$ is open and by Proposition~\ref{prop:closed} it is closed, so it is polite, and by Theorem~\ref{thm:polite} the protocol over it is volitionally grassroots.
\end{proof}

\iffull
\mypara{Openness and closure} Openness lets any person start an instance, and without permission.  A contract that gives one party of each instance a standing of its own is not open: two groups with no such party can never become connected.  Federated architectures fail this condition: each instance has its own server, and two instances remain two~\cite{lewis2026volitional}.

Closure leaves an instance requiring no global resource and no central authority: a registry that issues identifiers, a log that orders acts, or an authority that timestamps them takes part in a transaction without being a party to the contract.  Every act of a grassroots social contract is between its parties, so no third party takes part in one.  A smart contract that depends on a fact outside the blockchain obtains it from an \emph{oracle}, a third party whose report the contract cannot check~\cite{caldarelli2020oracle}; here a fact about the world enters as a party's signed assertion.  A smart contract requires a consensus protocol; a contract realised as volition-guarded transactions requires none where the parties transact between themselves~\cite{guerraoui2019consensus,lewis2023grassroots}, and where a collective decision is called for it constitutes its own consensus among the parties~\cite{keidar2025constitutional}.

\fi\full{}{A smart contract obtains a fact outside the blockchain from an oracle~\cite{caldarelli2020oracle}, a third party; a contract realised as volition-guarded transactions requires no consensus protocol where the parties transact between themselves~\cite{guerraoui2019consensus,lewis2023grassroots}, and a collective decision constitutes its own consensus among the parties~\cite{keidar2025constitutional}.
}
\subsection{The Social Graph Certified}\label{sec:sg-certified}

Befriend is an introductory act: for any two parties, at the binding sending Alice and Bob to them, every transaction it determines adds $\mathit{friend}(\Bob)$ at Alice and $\mathit{friend}(\Alice)$ at Bob, changing both states, and it is guarded in both roles, so each is a mutual recording of the two.  It is unobstructed: it requires nothing, and the only transactions adding $\mathit{friend}(\Bob)$ at Alice are those befriend determines at that binding.  Befriend is also the only schema adding a $\mathit{friend}$ atom and that atom's argument is a role, so $\mathit{friend}$ has traceable provenance.

Volition holds, by Definition~\ref{def:volition}: befriend is guarded in all its roles, sign has arity one, and of the other two, unfriend's role Alice requires $\mathit{friend}(\Bob)$ and forward's role Bob requires $\mathit{friend}(\Alice)$, each an atom of traceable provenance naming the other role.  The contract is therefore syntactically grassroots, and by Theorem~\ref{thm:grassroots} the protocol realising it is volitionally grassroots, so the platform realising the contract is a grassroots platform.

\iffull
\mypara{The instances} The introductory act is befriending and the edge each adds to the other is the $\mathit{friend}$ atom, so the $\mathit{friend}$-recording graph of a configuration is the friendship graph.  The instances of a run are the connected components of its $\mathit{friend}$-recording graph, taken undirected.  Every interaction of this contract is a befriend, an unfriend or a forward, and each of the three leaves or requires a $\mathit{friend}$ atom between its two participants at some configuration of the run.  Conversely, every $\mathit{friend}$ atom was added by a befriend, which is a mutual recording and so an interaction by Lemma~\ref{lem:recording-interaction}.  Forwarding adds an edge of another kind, from the recipient to the author, who has not transacted with them, as a coin does in Section~\ref{sec:cb-certified}.  The author is in the same instance nonetheless.  The set $\{\mathit{friend},\mathit{item},\mathit{sent}\}$ has traceable provenance: sign adds $\mathit{item}(x,\Alice,\Alice)$, which names its own role, and forward adds $\mathit{sent}(x,\Bob)$ at Alice and $\mathit{item}(x,a,\Alice)$ at Bob, whose person arguments are the roles Alice and Bob and the party variable $a$, which the required $\mathit{item}(x,a,f)$ names; so Lemma~\ref{lem:grounded} applies.

\fi
\iffull
\subsection{The Currency Certified}\label{sec:cb-certified}

The swap is an introductory act: at the binding sending $u$ to Alice and $v$ to Bob every transaction it determines leaves Alice holding a coin Bob issued and Bob a coin Alice issued, changing both states, and it is guarded in both roles, so each is a mutual recording of the two, the mutual credit line.  It is unobstructed: it forbids nothing, and at that binding each role requires a coin of its own issue, which mint adds, mint being of arity one, guarded in its role, and requiring and forbidding nothing.

\textcent{} has traceable provenance.  Mint adds \textcent$(\Alice)$ at Alice and pay adds \textcent$(\Bob)$ at Bob, each naming the role that holds it; the swap adds in each role a coin the other role requires; and redeem adds \textcent$(r)$ at Alice, which role Bob requires, and \textcent$(\Bob)$ at Bob, naming its own role.  Volition holds, by Definition~\ref{def:volition}: mint has arity one, the swap is guarded in both its roles, and pay and redeem are guarded at Alice only and require there \textcent$(\Bob)$, an atom of traceable provenance having Bob among its arguments.  The contract is therefore syntactically grassroots, and by Theorem~\ref{thm:grassroots} the protocol realising it is volitionally grassroots, so the platform realising the contract is a grassroots platform.

\iffull
\mypara{The origin of a record} A coin and an item are alike: each is an atom naming a person who is not a participant, put in a state by an act of two parties.  A $\mathit{friend}$ edge can only have been added by the two people it joins; a coin or an item can be handed on, so an edge may point at a person its source has never transacted with.  Traceable provenance covers both.  An edge is added either to a party to the act that adds it or to a person another party to that act records already, so no act ever leaves an edge whose ends lie in different instances, which Lemma~\ref{lem:grounded} proves by induction on the run.  For coins the chain is the passage of the coin from its issuer, and conservation of money~\cite{lewis2026volitional} is the special case, here an invariant of the schemas rather than one proved of the runs.  Requiring an atom to carry the parties that passed it on is the same condition written as a clause of the contract rather than as one on the schemas~\cite{golike2026digital}.  The two differ in one respect: pay deletes the coin from the payer, and forward keeps the item.

\fi
\fi

%% file: sections/08-modalities.tex
\section{Clauses and Their Modalities}\label{sec:modalities}

\iffull
The schemas also determine, for each clause, whether a party can breach it and what record of it the contract leaves in another party's state.  The two contracts show it.  A clause may be one no act of the contract can breach, one a party may breach with the contract recording the breach, or one whose performance lies outside the contract and of which it records only the undertaking.  The question is about every clause, not only about those that describe no act.  The three are what a realisation does with a clause, and are not kinds of obligation: a friendship kept on paper can be entered against a person's will, and under the schemas of Section~\ref{sec:sg-schemas} it cannot.
\fi

Throughout this section $C$ is a syntactically grassroots contract.  A party joins the platform realising $C$ by signing an \emph{undertaking} to keep $C$ towards all its other participants, which it presents to every participant it interacts with (Section~\ref{sec:signing}).

\begin{definition}[Clause]\label{def:clause}
A \emph{clause} of a contract $C$ is a proposition about the conduct of its parties, and is \emph{kept} or breached by that conduct; a party the clause names as answerable is its \emph{obligor}.  The runs of the protocol realising $C$ \emph{determine} the clause if whether it is kept is a function of the run.  A transaction of $R_C$ \emph{may breach} the clause if there are circumstances in which its occurrence puts a party in breach of it.
\end{definition}

\begin{definition}[Witnessed]\label{def:witnessed}
A volition-guarded transaction $(d \rightarrow d',Q')$ over participants $Q$ is \emph{witnessed against} $a \in \Pi$ if the recording graph of $d'$ has an edge $b \rightarrow a$ for some $b \in Q$ with $b \neq a$.  A schema is \emph{witnessed against} a role if every transaction it determines is witnessed against the person in that role.
\end{definition}

By Lemma~\ref{lem:records} that edge is an atom naming $a$ in the state of another party.  A record serves as evidence when it names the party answerable and is held by someone else.

\begin{definition}[Enforced, Attested, Undertaken]\label{def:modalities}
Let $\varphi$ be a clause of $C$, with obligor $o$ where it has one.  Then $\varphi$ is
\begin{enumerate}
    \item \emph{enforced} by $C$ if the runs of the protocol realising $C$ determine it and no transaction of $R_C$ may breach it;
    \item \emph{attested} by $C$ if some transaction of $R_C$ may breach it and every transaction that may is witnessed against $o$;
    \item \emph{undertaken} by $C$ if no transaction of $R_C$ may breach it, the runs do not determine it, and some transaction of $R_C$ is witnessed against $o$ or $o$'s undertaking is held by another party.
\end{enumerate}
\end{definition}

\iffull
The three are exclusive and not exhaustive.  Two kinds of clause are of none.  One that some transaction may breach without witnessing against the obligor can be breached invisibly.  One that the runs neither determine nor record is owed by no one identifiable.  A realisation does nothing with either.

\fi
A schema \emph{names} its role $\pi_j$ at another role $\pi_i$ if an atom added at $\pi_i$ has $\pi_j$ among its arguments.

\begin{proposition}[Witnessing]\label{prop:witnessing}
A schema is witnessed against every role it names at another role.
\end{proposition}

\iffull
\begin{proof}
Let $\tau$ name $\pi_j$ at $\pi_i$ by an atom $\varphi$ added at $\pi_i$, let $\beta$ be a binding for $\tau$, and put $a := \beta(\pi_j)$ and $b := \beta(\pi_i)$, which differ as $\beta$ is injective on roles.  By Definition~\ref{def:instance} every transaction determined by $\tau$ at $\beta$ leaves $b$ with $\beta(\mathord{+}(i))$ adjoined, so the state of $b$ after it contains $\beta(\varphi)$, an atom having $a$ among its arguments, and the recording graph after it has the edge $b \rightarrow a$ by Lemma~\ref{lem:records}.
\end{proof}
\fi
A predicate $e$ is \emph{guarded} in $C$ if whenever a schema of $C$ adds an $e$-atom, the atom names only roles of that schema that will the act.  The \emph{consent clause} for $e$ says that a party's state holds an $e$-atom only if those the atom names willed the act that added it.

\begin{proposition}[Enforcement by guarding]\label{prop:enforced}
If $e$ is guarded in $C$, the consent clause for $e$ is enforced by $C$.
\end{proposition}

\iffull
\begin{proof}
Let $r$ be a finite safe run of the protocol realising $C$ and let an $e$-atom occur in the state of a party at some configuration of $r$.  The initial state is empty, so some transition of $r$ added it, and by Definitions~\ref{def:realisation} and~\ref{def:instance} that atom is $\beta(\varphi)$ for a binding $\beta$ of a schema $\tau$ of $C$ and some $\varphi \in \mathord{+}(i)$ with predicate $e$.  By hypothesis every argument of that atom is a guarding role of $\tau$, so every person among the atom's arguments is $\beta$ of a guarding role and hence in the guard of the transaction taken; by Definition~\ref{def:enabled} each held its class in its volitional state at that point, which is to say willed the act.
\end{proof}
\fi

\mypara{The clauses of the \full{two contracts}{contract}} Clauses~1 and~2 of Section~\ref{sec:sg} are \emph{enforced}, by Proposition~\ref{prop:enforced} and by unfriend deleting in both roles.  Clause~5 of Section~\ref{sec:sg} is \emph{undertaken}, forward adding $\mathit{item}(x,a,\Alice)$ at Bob and so being witnessed against the party that passed the item on\full{; and so is clause~5 of Section~\ref{sec:cb}, a mutual credit line adding \textcent$(\Bob)$ at Alice and \textcent$(\Alice)$ at Bob and so being witnessed against both of them.  The clause that a coin enters a holding only by its issuer's will is of no kind, being breachable by pay without pay naming the payer}{}.  The remaining clauses permit acts and oblige nothing, so no transaction may breach them and every run keeps them: they are enforced.

\iffull
\mypara{Attestation and undertaking} The two require the same thing of the schemas, that the act naming the obligor leave the record in another party's hands, and differ only in whether the breach is itself an act of the contract.  The evidence depends on the schemas and the modality on the clause.  A performance promised for later is never attested, because failing to perform is not an act; a warranty given in the course of an act is attested exactly when that act is witnessed against the party giving it.

\fi
\iffull
Witnessing therefore depends on whether an effect names the party performing the act, and not only the parties the atom concerns.  Unfriend adds no atom, so no clause against ending a friendship at will can be attested under the contract of Section~\ref{sec:sg}; pay adds a coin naming its issuer and not its payer, so no clause against the payer can be attested under the contract of Section~\ref{sec:cb}.  A contract whose breaches are to be attested must have acts whose effects name their author.  Signing provides it: an act taken as a digital speech act, an utterance signed with the key of the party taking it, leaves an atom naming that party, and the signed record is non-repudiable~\cite{cardelli2020digital,golike2026digital}.

\fi
\mypara{Performance outside the system} \full{Clause~5 of Section~\ref{sec:cb} and the warranty in clause~4 of Section~\ref{sec:sg} are performed outside the schemas, since selling at an advertised price is no act of the contract and neither is knowing a person, and they have different modalities.}{The warranty in clause~4 of Section~\ref{sec:sg} is performed outside the schemas, knowing a person being no act of the contract, and it is \emph{attested}: forwarding is an act, a forward by a party that does not personally know the party it received the item from makes the warranty false, and every forward adds $\mathit{item}(x,a,\Alice)$ at Bob, so by Proposition~\ref{prop:witnessing} it is witnessed against Alice.}
  \iffull
The first is \emph{undertaken}: no transaction may breach it, no run determines whether the issuer honours it, and a mutual credit line is witnessed against the issuer, so the coin in another party's holding is the record of the undertaking.  The second is \emph{attested}: forwarding is an act, a forward by a party that does not personally know the party it received the item from makes the warranty false, and every forward adds $\mathit{item}(x,a,\Alice)$ at Bob, so by Proposition~\ref{prop:witnessing} it is witnessed against Alice.  The three modalities therefore do not divide by whether performance is digital.
\fi

%% file: sections/09-signing.tex
\section{Signing the Contract}\label{sec:signing}

\iffull A person joins a global platform by signing an End-User Licence Agreement with its operator.  \fi A person joins a grassroots platform by signing a one-sided \emph{undertaking} to keep its grassroots social contract towards all its other participants, and presents it to every participant it interacts with\full{; the code realising the contract, or the platform hosting it~\cite{shapiro2026gsg}, enforces this, refusing an interaction with a participant that has not presented one.  The undertaking is a speech act, signed with the key of the party giving it and non-repudiable~\cite{cardelli2020digital,golike2026digital}.  A person is asked to sign on the invitation of a participant, whose signed invitation names the contract and the invitee, and the invitee presents its undertaking first to the inviter~\cite{shapiro2026gsg}.}{, as the code realising the contract or the platform hosting it~\cite{shapiro2026gsg} enforces.}

\iffull
Presentation is not part of the protocol: the code carries the undertaking in the messages of an interaction, and it leaves no atom in a state.  Throughout this section $C$ is a syntactically grassroots contract, $\calF$ the protocol realising it and $Q \subset \Pi$; write $p \rhd q$, for $p, q \in Q$, when $p$ has presented its undertaking to $q$.

\begin{definition}[Presenting run]\label{def:presenting}
A run of $\calF(Q)$ with a relation $\rhd$ on $Q$ is \emph{presenting} if $p \rhd q$ and $q \rhd p$ for every interaction of the run between $p$ and $q$.
\end{definition}

The code realising $C$, or the platform hosting it, keeps every run presenting.

\begin{proposition}[Every participant has signed]\label{prop:signed}
In a presenting run, every agent of an instance of two or more agents has presented its undertaking to another agent of that instance.
\end{proposition}

\begin{proof}
An instance is a connected component of the interaction graph of the run (Definition~\ref{def:coalescence}), so an agent of an instance of two or more agents has an edge to another agent $q$ of it, put there by an interaction between the two; the run is presenting, so the agent has presented its undertaking to $q$.
\end{proof}

By Section~\ref{sec:modalities} a record is evidence when it names the party answerable and is held by someone else; the undertaking names its signer, and by the proposition another participant holds it.  Condition (5) of Section~\ref{sec:intro} is therefore enforced by the code.
\fi
\full{}{Hence every participant of an instance of two or more has signed it and another participant holds it: condition (5) of the abstract is enforced by the code.}

\full{\mypara{The code} }{}Enforced clauses are enforced by the code the parties run\full{: a breach of one is impossible with a correct implementation of the contract (Section~\ref{sec:modalities}), so a breach requires a party to run other code}{}.  The code is certified by its \emph{creator}, a person\full{, natural or legal}{}, and by its \emph{compiler}\full{.  Each certification is a signed warranty carried with the code: the creator's, that the code realises the contract, carrying the contract's schemas and the checker's verdict that they are syntactically grassroots; the compiler's, that the compiled code is the creator's.  The certifier is not necessarily a party, so the certifications are not clauses of the contract and have no modality under Definition~\ref{def:modalities}; a breach of either is provable from the certified code and the contract by anyone holding both}{, neither necessarily a party}.  In its undertaking a party also promises to run the certified code\full{.  This clause is \emph{undertaken}: no act of the contract can breach it, no run shows whether a party keeps it, and the undertaking, in another party's hands, is its record.}{; this clause is \emph{undertaken}, and the undertaking, in another party's hands, is its record.}\full{  Where the platform hosting the code attests it, every interaction opening with proof that both parties run the certified code, as in the Grassroots Super-App~\cite{shapiro2026gsg}, the clause is \emph{enforced} by the platform.}{  Where the hosting platform attests the code~\cite{shapiro2026gsg}, the clause is \emph{enforced} by the platform.}

\mypara{Electronic signatures} An undertaking signed with a party's key is an electronic signature.  In the United States a signature, contract or record may not be denied legal effect, validity or enforceability solely because it is in electronic form~\cite{esign2000}; in the European Union an electronic signature shall not be denied legal effect and admissibility as evidence solely on the grounds that it is in electronic form~\cite{eidas2014}; and in the United Kingdom an electronic signature, and its certification, are admissible in evidence on the authenticity and integrity of the data signed~\cite{eca2000}.  A grassroots social contract therefore requires no new doctrine of contract formation: its parties freely accept and sign it, as condition (5) requires, and its realisation keeps signed records of the acts that follow.

%% file: sections/10-related-work.tex
\section{Related Work}\label{sec:related}

\full{\mypara{Contracts as formal objects} The formal study of contracts in this field is dominated by the contract-as-automaton family.  \emph{Contract automata}~\cite{azzopardi2016contract} give the deontic modalities an operational semantics over the acts of named interactive parties; Flood and Goodenough~\cite{flood2022automaton} render a loan agreement as a deterministic finite automaton whose states are performing, delinquent and default and whose transitions are drawn from an alphabet of events; the L4 domain-specific language~\cite{watt2023l4} has an executable semantics in which deontic clauses with deadlines become states and transitions; and timed contract automata~\cite{chircop2022timed} add quantitative time.  Rule-based encodings represent contract clauses as defeasible deontic rules with violation and reparation chains~\cite{governatori2004ruleml}, and imperative and declarative smart contracts running on a distributed ledger have been compared~\cite{governatori2018legal}.}{\mypara{Contracts as formal objects} The formal study of contracts in this field is dominated by the contract-as-automaton family --- contract automata~\cite{azzopardi2016contract}, a loan agreement as a finite automaton~\cite{flood2022automaton}, the L4 language~\cite{watt2023l4}, timed contract automata~\cite{chircop2022timed}, and encodings as defeasible deontic rules~\cite{governatori2004ruleml} and the comparison of imperative with declarative smart contracts~\cite{governatori2018legal}.  All fix the parties when the contract is formed, keep one contract state, and assume an alphabet of events, a clock or a ledger, where our conditions quantify over the family of systems a contract induces on every finite set of people, need a local state per party, and forbid a participant that is not a party.  Symboleo~\cite{sharifi2020symboleo,parvizimosaed2022symboleo} and Stipula~\cite{crafa2022stipula} are the nearest to ours and make the same three commitments; Catala~\cite{merigoux2021catala} does the corresponding work for statute.  Befriending cannot be written in Stipula, which allows joint consent only in the constructor.  We claim no new language: the conditions are on the four things a schema carries, which Stipula's functions and Symboleo's obligations also carry, given a semantics in which each party holds its own state.}

\full{Three commitments are common to all of these, and this paper makes none of them.  The parties are fixed when the contract is formed, so the object of study is one instance among a known set of parties, where our conditions quantify over the family of systems a contract induces on every finite set of people.  There is one contract state, the automaton's state, the rule base's facts or the ledger, where the definition of \emph{records}, and therefore of what it is for two parties to be connected, needs a local state per party and a local-states function monotone in the set of people.  And each requires something closure forbids: an alphabet of events every party is assumed to observe alike, a clock, or a ledger.}{}

\full{Two languages outside this venue are the nearest to ours.  Symboleo~\cite{sharifi2020symboleo,parvizimosaed2022symboleo} specifies a legal contract as a set of obligations and powers over an ontology of roles, assets and events, with a semantics of logical axioms on statecharts, and properties verified in temporal logic; its purpose is monitoring, and roles are assigned to parties during each contract execution.  Stipula~\cite{crafa2022stipula} is a domain-specific language with legal constructs, agreement, permissions, obligations and violations, whose contracts are state machines over parties, fields and linear assets; its \emph{agreement} operator is the contract's constructor, fixing the full list of parties, and thereafter each function is invoked by one named party.  Catala~\cite{merigoux2021catala} does the corresponding work for statute rather than contract, encoding the base-case and exception structure of legislation in default logic.  Both contract languages make the three commitments above: a Stipula configuration carries a global clock, and Symboleo presupposes a monitor, which is a participant that is not a party.}{}

\iffull
A concrete case is befriending, which cannot be written in Stipula: it is an act that two parties must both will, available repeatedly between any two of them, and Stipula allows joint consent only in the constructor, so each befriending would be the formation of a new contract rather than a clause of a standing one.

We claim no new language.  The conditions of Section~\ref{sec:conditions} are on the four things a schema carries.  Stipula's function declarations name an invoking party and a precondition, and Symboleo attributes obligations to a debtor and a creditor, so the four could be extracted from either given a semantics in which each party holds its own state.

\fi
\iffull
\mypara{Norms and normative systems} The formal study of what a contract obliges uses deontic and defeasible logic: contracts are represented as defeasible rules with deontic operators and compliance is checked against them~\cite{governatori2018practical}.  Normative multi-agent systems concern how norms are operationalised, detected, and brought into conflict~\cite{boella2006introduction}.  A volition guard~\cite{lewis2026volitional} is a condition of enactment rather than a deontic operator: an act does not occur unless the people it obliges are willing.  Where a clause of the contract describes no act the protocol carries out, Section~\ref{sec:modalities} says what the realisation does with it: record its breach, or record the undertaking and the party who gave it.

\mypara{Juridical acts and declarative power} What we call an act is Sartor's \emph{result-declaration}, one of ``acts intended to produce legal determinations''~\cite{sartor2006fundamental}, and the fullest model of it is Hage's~\cite{hage2011juridical}, which analyses juridical acts as intentional changes in a world of law furnished with entities, facts and rules.  A multi-party act willed by its parties is not new either.  Declarative power is ``the capacity of the power-holder of creating normative positions, involving other agents, simply by `proclaiming' such positions''~\cite{gelati2004normative}, and its authors observe that a declarative power exercised jointly by several parties, with the consent of each, is a contract~\cite{gelati2002declarative}.

We differ from that tradition on where the effect of an act comes from.  There a proclamation has effect only if an institution provides for it: ``when an agent $j$ proclaims $A$, $j$ brings it about that $A$ only if the concerned institution $s$ provides for this result''~\cite{gelati2002declarative}; in norm-governed institutions ``designated agents are \emph{empowered} to create particular kinds of states of affairs''~\cite{jones1996institutionalised}; counts-as statements ``hold only with respect to a context''~\cite{grossi2005countsas}; and an artificial institution ``presupposes an agreement on an unambiguous definition of a set of concepts''~\cite{fornara2008institutions}.  A grassroots social contract has no such institution, being subject to no external authority.  The effect of an act is therefore the change it makes in the parties' own states, and the condition of its effectiveness is their will rather than a conferred power.  An act schema is a result-declaration with the institution taken out: its precondition is on each participant's own local state rather than on an institutionally conferred power, and its guard takes the place of authorisation.

\mypara{Standard-form contracting} A grassroots social contract has no drafting party, so there is no party whose terms are policed and no party from whom disclosure is required; the defects the literature finds in the arrangement between a person and a proprietor are absent because the proprietor is absent.

The End-User Licence Agreement is a \emph{contract of adhesion}.  Radin~\cite{radin2013boilerplate} treats mass-market boilerplate as the deletion of rights that the legal system otherwise confers, and as a matter for the rule of law rather than for consent only.  Kim~\cite{kim2013wrap} treats the digital instances of the form --- the terms accepted by clicking, browsing, or installing --- and what assent means when the act of assent is a condition of access.  Ben-Shahar and Schneider~\cite{benshahar2014more} address the remedy most often proposed, mandated disclosure, and find that it fails: the terms are not read, and requiring more of them is not read either.

This literature takes the bilateral form as given; its questions are what may be done within the form: whether assent was real, which terms a court should refuse, what a drafter must disclose.

\mypara{Terms a machine can evaluate} Surden~\cite{surden2012computable} separates the terms of a contract a machine can evaluate from those it cannot; work on smart legal contracts separates the operational parts of an agreement, which code executes, from the rest, which stays in prose~\cite{clack2016templates}; and in security, Schneider~\cite{schneider2000enforceable} characterises the policies an execution monitor can enforce.  Each divides a contract's terms in two, by what a realisation can do with them.  Section~\ref{sec:modalities} divides clauses in three, distinguishing a clause no act of the contract can breach from a clause a party may breach while the contract records the breach against them.  A two-way division has no room for the second, and it is the class the law needs: a remedy requires a provable breach and an identifiable obligor.

In the monitoring literature, compliance is computed by monitors that ``receive inputs from \emph{trusted} observers''~\cite{modgil2015monitoring}, and where the timestamps cannot be trusted the monitor is made robust rather than removed: compliance ``is typically computed with respect to timed event traces with event timestamps assumed to be perfect'', and the remedy offered is a semantics for compliance when they are not~\cite{cambronero2017timed}.  A monitor is a participant that is not a party, and closure forbids one.  Here the record is held by the parties themselves, and what makes it evidence is Definition~\ref{def:witnessed}: it names the party answerable and it lies in someone else's hands.

\mypara{Code as a regulator} Lessig~\cite{lessig1999code} argued that architecture regulates conduct as law, markets and norms do, and that in the digital realm the architecture is code.  A grassroots social contract is not code: the conditions of Section~\ref{sec:conditions} concern who may take part in an act and whose will is required for it, rather than what the code makes a party do, and they are checked on the text rather than on the code.  The digital social contract~\cite{cardelli2020digital} identifies the contract with code: the contract is a program, and a party can behave only according to it; the smart contract literature~\cite{de2021smart} does the same on a blockchain.

\mypara{Self-governance without an external authority} Ostrom~\cite{ostrom1990governing} documented communities that devise, adopt and enforce their own rules over shared resources, without an owner and without a state administering them, and identified the conditions under which such arrangements endure.  A grassroots social contract makes the same claim in the digital realm, that a community constitutes itself, and openness is its formal counterpart: any two people can come to record each other by acts of their own, so any person may start an instance and instances coexist.  Commons scholarship concerns how a resource is governed in common, for physical resources~\cite{ostrom1990governing} and for knowledge and networked production~\cite{hess2007understanding,benkler2006wealth}; a grassroots social contract governs a relationship rather than a resource, and applies where nothing is held in common.

Platform regulation binds the operator: the General Data Protection Regulation~\cite{gdpr2016} imposes duties on the party that processes personal data, and so presupposes it.  Where a community's platform is constituted from its members' own devices~\cite{shapiro2023grassrootsBA,shapiro2024grassroots,shapiro2025characterising}, there is no operator to bind.

\fi
\full{}{\mypara{Norms and normative systems} Contracts are represented as defeasible rules with deontic operators, against which compliance is checked~\cite{governatori2018practical}, and normative multi-agent systems concern how norms are operationalised and brought into conflict~\cite{boella2006introduction}; a guard~\cite{lewis2026volitional} is a condition of enactment rather than a deontic operator.

\mypara{Juridical acts and declarative power} What we call an act is the \emph{result-declaration} of Sartor~\cite{sartor2006fundamental}, modelled by Hage~\cite{hage2011juridical}, and a declarative power exercised jointly is a contract~\cite{gelati2004normative,gelati2002declarative}.  There a proclamation has effect only if an institution provides for it~\cite{gelati2002declarative,jones1996institutionalised,grossi2005countsas,fornara2008institutions}; a grassroots social contract has no institution, so an act's effect is the change in the parties' own states, and its guard takes the place of authorisation.

\mypara{Standard-form contracting} The End-User Licence Agreement is a \emph{contract of adhesion}, treated by Radin~\cite{radin2013boilerplate}, Kim~\cite{kim2013wrap} and Ben-Shahar and Schneider~\cite{benshahar2014more}; a grassroots social contract has no drafting party, so the questions of that literature do not arise.

\mypara{Terms a machine can evaluate} Surden~\cite{surden2012computable}, smart legal contracts~\cite{clack2016templates} and enforceable security policies~\cite{schneider2000enforceable} each divide a contract's terms in two by what a realisation can do with them; Section~\ref{sec:modalities} divides them in three, and the third class, a breach recorded against its obligor, is the one a remedy needs.  Compliance monitors receive inputs from trusted observers~\cite{modgil2015monitoring,cambronero2017timed}; a monitor is a participant that is not a party, and closure forbids one.

\mypara{Code as a regulator} Lessig~\cite{lessig1999code} argued that architecture regulates as code, and the digital social contract~\cite{cardelli2020digital} and smart contracts~\cite{de2021smart} make the contract a program; a grassroots social contract is not code, its conditions being checked on the text.

\mypara{Self-governance} Ostrom~\cite{ostrom1990governing} and commons scholarship~\cite{hess2007understanding,benkler2006wealth} concern communities governing a resource without an owner or a state, and platform regulation~\cite{gdpr2016} binds an operator; a grassroots social contract governs a relationship, and a platform constituted from its members' own devices~\cite{shapiro2023grassrootsBA,shapiro2024grassroots,shapiro2025characterising} has no operator.}

\full{\mypara{Grassroots systems} The notion of a grassroots protocol is due to~\cite{shapiro2023grassrootsBA} and was recast in terms of interleaving and coalescence in~\cite{lewis2026volitional}, which also introduced volition-guarded transactions and volitionally grassroots protocols.  Grassroots platforms have been specified for social networking~\cite{shapiro2023gsn,shapiro2026gsg} and for personal cryptocurrencies~\cite{shapiro2024gc,lewis2023grassroots,shapiro2026bonds}.  Three things are added here.  The first is politeness, a condition on a protocol's transactions, which Theorem~\ref{thm:polite} proves sufficient for the protocol to be grassroots, so no proof about runs is needed.  The second is the step before the specification: the legal text, and decidable conditions on it under which politeness is established from the acts rather than from the protocol.  The third is what a realisation does with a clause of the contract, whether or not the contract provides an act for it, which depends on the same schemas.}{\mypara{Grassroots systems} Grassroots protocols are due to~\cite{shapiro2023grassrootsBA} and were recast with volition-guarded transactions in~\cite{lewis2026volitional}; platforms have been specified for social networking~\cite{shapiro2023gsn,shapiro2026gsg} and for personal cryptocurrencies~\cite{shapiro2024gc,lewis2023grassroots,shapiro2026bonds}.  Added here are politeness, the legal text with decidable conditions on it, and what a realisation does with a clause.}

\iffull
\mypara{Neighbouring instruments} Four instruments resemble a grassroots social contract, and each differs from it in what it does rather than in how it is drafted.  A \emph{licence} states what one party may do with something another party is entitled to: the General Public License~\cite{fsf2007gpl} presupposes the copyright it licenses, grants freedoms on conditions, and binds a taker because they would otherwise infringe.  A grassroots social contract licenses nothing, and its obligations are owed by each party to every other and derive from no party's prior entitlement, so a person who has never held any right that the others could infringe is bound as the rest are.  \emph{Bylaws} bind the members of an association symmetrically and are adopted on joining, which makes them the closest of the four, but they presuppose an association: an entity with legal personality, organs that act for it, and a membership register.  A \emph{code of conduct} states standards of behaviour, and its sanction is exclusion administered by whoever runs the forum, which is the party a grassroots social contract does not have; \emph{terms of membership} presuppose the same party, in its capacity as the one who admits.
\fi
\full{}{\mypara{Neighbouring instruments} A licence~\cite{fsf2007gpl} presupposes an entitlement, bylaws an association, a code of conduct and terms of membership whoever runs the forum; a grassroots social contract presupposes none.}

%% file: sections/11-conclusion.tex
\section{Conclusion}\label{sec:conclusion}

Proudhon's conditions on a social contract, with two more, formalised as conditions on its acts, give it the property he envisioned for it: any two people may come into relation without anyone's leave, two communities that adopted the contract independently can later become one, and every participant in an act is a party.  A grassroots social contract is one meeting the conditions, and its realisation has all three.

\iffull
Such a contract has two artefacts, a legal text and a protocol realising it, and the grassroots literature has had only the second.  This paper translates the first into the second.  A clause of the text describes an act, and the acts become act schemas; the schemas compile into volition-guarded transactions over a local-states function in which a state records exactly the people its atoms name; and three conditions on the schemas make a contract \emph{syntactically grassroots}, which is enough for the protocol realising it to be volitionally grassroots.  The compilation and the conditions are algorithmic and the conditions are decidable, so a contract can be certified from its text, and the only step left to judgement is the first, rendering a text as schemas.

The same schemas determine what a realisation does with a clause, whether or not the clause describes an act: leave it unbreachable, record its breach against the party in breach, or record the undertaking and who gave it.  This depends on one thing: whether the act that names the party answerable leaves the record in another party's hands.  Attestation and undertaking both require that of the schemas, and differ only in whether the breach is itself an act of the contract.

Both contracts are certified from their schemas.  Who holds a record of whom is a directed graph, and an instance is a connected component of the undirected graph of interactions.  In the social graph the two coincide: befriending is the introductory act, and a friend edge can only have been added by the two people it joins.  A coin can be handed on and so can a signed item, so an edge may point at a person its source has never transacted with; an edge is added either to a party to the act that adds it or to a person another party to that act records already, and that condition on the schemas is enough for no edge ever to join two instances.  For coins it is conservation of money.

The conditions also leave no participant that is not a party.  Openness lets any person start an instance, and without permission.  Closure leaves an instance requiring no registry that issues identifiers, no log that orders acts, and no authority that timestamps them, each of which would take part in an act without being a party to the contract.  A community organised in this way requires no permission to form, no authority to operate, and no proprietor to govern it.

Three directions remain.  The conditions are sufficient and not necessary, and a characterisation would say which contracts are grassroots.  Which further invariants of a contract follow from its schemas, as conservation of money does, is open.  And the four things a schema carries are present, under other names, in the contract specification languages of Section~\ref{sec:related}; giving one of them a semantics in which each party holds its own state would let the conditions be checked on contracts already written in it.
\fi

%% file: sections/12-appendix.tex
\section{The Framework of Volition-Guarded Multiagent Atomic Transactions}\label{app:framework}

This appendix states formally what Section~\ref{sec:transactions} recalls in brief, so that the statements and proofs of Section~\ref{sec:polite} can be read without recourse to another paper.  Everything in it is from~\cite{lewis2026volitional}, and the proofs of its results are there and are not repeated.

Earlier work introduced multiagent transition systems~\cite{shapiro2021multiagent}, grassroots protocols and platforms~\cite{shapiro2023grassrootsBA}, and their definition via multiagent atomic transactions~\cite{shapiro2025atomic}, which describe the behaviour of machines but not of the people operating them.  A volition-guarded transaction is a machine transaction guarded by the volitions of some, all, or none of the people whose machines participate; a person's volitional state is a set of equivalence classes of machine transactions they are willing their machine to take part in.  The social graph illustrates the two extremes: befriending is guarded by both parties, while unfriending is guarded by either.  A person may freely change their volitional state; an equivalence class is removed from every agent's volitional state when a machine transaction in that class is taken, the will having been fulfilled.

\subsection{Agents, Machines and Volition-Guarded Transactions}\label{app:vgt}

We assume a potentially infinite set of \emph{agents} $\Pi$, an agent being a person operating a machine, but consider only finite subsets of it, so when referring to a particular set of agents $P \subset \Pi$ we assume $P$ to be nonempty and finite.  We use $\subset$ to denote the strict subset relation and $\subseteq$ when equality is also possible, and use $p\ne q \in P$ as a shorthand for $p\in P \wedge q\in P \wedge p\ne q$.  As standard, we use $S^P$ to denote the set of all total functions from $P$ to $S$, and if $c\in S^P$ we use $c_p$ to denote the value of $c$ at $p\in P$.

\begin{definition}[Machine State, Configuration, Transaction, Volition-Guarded Transaction~\cite{lewis2026volitional}]\label{def:mt}
Given an arbitrary set $S$ of \emph{machine states}, with a designated \emph{initial state} $s_0 \in S$, and agents $Q \subset \Pi$, a \emph{machine configuration} over $Q$ is a member of $S^Q$, and a \emph{machine transaction} over \emph{participants} $Q$ is a pair $c\rightarrow c' \in (S^Q)^2$ such that $c\ne c'$; it is \emph{unary} if $|Q|=1$.  Given such a machine transaction $t$, a \emph{volition-guarded multiagent atomic transaction} over $t$---henceforth, \emph{volition-guarded transaction}---is a pair $(t,Q')$ where $Q'\subseteq Q$ are its \emph{guards}.
\end{definition}

Machine transactions are atomic and asynchronous~\cite{shapiro2021multiagent}: they can be carried out by their participants at any time, regardless of the states of non-participants.  Participants include both active agents, whose state changes, and stationary agents, whose state is a precondition but does not change.  Volition-guarded transactions can be carried out only if their guards $Q'\subseteq Q$ are willing, and do not distinguish between agents that initiate a transaction and those willing to take part in it.  When we say a transaction is guarded by $\{p,q\}$, both must be willing; when we say it is guarded by either $p$ or $q$, we mean there are two volition-guarded transactions over the same machine transaction, $(t,\{p\})$ and $(t,\{q\})$, so that either person's volition suffices.

Distinct machine transactions can represent the same act in different configurations; a transaction equivalence relates them.

\begin{definition}[Transaction Equivalence~\cite{lewis2026volitional}]\label{def:equivalence}
Given a set of machine transactions $R$, a \emph{transaction equivalence} is an equivalence relation $\sim$ on $R$ such that $t \sim t'$ implies $t$ and $t'$ have the same participants, and $R/{\sim}$ is countable.  We write $[t]$ for the equivalence class of $t$ under $\sim$.
\end{definition}

All befriend transactions of $p$ and $q$, differing only in the configurations in which they occur, form one equivalence class.  A person wills a class, and taking one member of it fulfils that will.

\begin{definition}[Agent State and Configuration~\cite{lewis2026volitional}]\label{def:agent-state}
Given agents $P$, states $S$ with initial state $s_0$, a set of machine transactions $T$ each over its own participants $Q\subseteq P$ and $S$, and equivalence $\sim$ on $T$, an \emph{agent state} is a pair $(V,m)\in \calA = (2^{T/\sim}  \times S)$ where  $V$ is its \emph{volitional state} and $m \in S$ its  \emph{machine state}.  The \emph{initial agent state} is $(\emptyset,s_0)$.  An \emph{agent configuration} $c$ over $P$, $S$, $T$, and $\sim$ is a member $c\in \calA^P$ in which $c^v_p \subseteq (T/{\sim})_p$ for every $p\in P$, where $(T/{\sim})_p$ denotes the classes in $T/{\sim}$ in which $p$ is a participant; we write $c^v_p$ for the volitional state and $c^m_p$ for the machine state of agent $p$ in $c$.
\end{definition}

\begin{definition}[Volitional Transaction~\cite{lewis2026volitional}]\label{def:vmat}
Given agents $P$, states $S$, machine transactions $T$ over $P$ and $S$, and equivalence $\sim$ on $T$:
\begin{enumerate}
    \item A \emph{volition change of agent $p\in P$} is a pair $c\rightarrow c'$ of agent configurations over $\{p\}$, $S$, $T$, and $\sim$ such that $c^v_p, c'{^v_p} \subseteq (T/{\sim})_p$ and $c^v_p \ne c'{^v_p}$, and $c^m_p = c'{^m_p}$.
    \item A \emph{volitional machine transaction} induced by a volition-guarded transaction $(t,Q')$, for some $t= (d\rightarrow d')\in T$ over $Q\subseteq P$ with guards $Q'\subseteq Q$, is a pair $c\rightarrow c'$ where $c\ne c'$ are agent configurations over $P$, $S$, $T$, and $\sim$ such that $[t]\in c^v_q$ for every $q\in Q'$; $c^m_p = d_p$ and $c'{^m_p} = d'_p$ for every $p\in Q$; $c^m_p = c'{^m_p}$ for every $p\in P\setminus Q$; and $c'{^v_p} = c^v_p \setminus \{[t]\}$ for every $p\in P$.
    \item A \emph{volitional transaction} is a volition change or a volitional machine transaction.
\end{enumerate}
\end{definition}

When a volitional machine transaction induced by $(t,Q')$ is taken, the class $[t]$ is removed from every agent's volitional state: the will is fulfilled by any equivalent transaction.  A person may independently change their volitional state via volition changes, which may add or remove classes; beyond these, the framework removes a class from $c^v_p$ only upon fulfilment.

\subsection{Transition Systems and the Induced Volitional System}\label{app:vmts}

\begin{definition}[Transition System, Computation, Run, Safe, Live, Correct~\cite{lewis2026volitional}]\label{def:ts}\label{def:liveness}
A \emph{transition system} is a tuple $TS=(S,s_0,T,{\sim})$, where:
\begin{enumerate}
    \item $S$ is an arbitrary non-empty set, referred to as the set of \emph{states}.
    \item Some $s_0\in S$ is the designated \emph{initial state}.
    \item $T\subseteq S^2$ is a set of \emph{correct transitions over} $S$, where each transition $t\in T$ is a pair $(s,s')$ of non-identical states $s\ne s'\in S$, also written as $t=s\rightarrow s'$.
    \item $\sim$ is a \emph{partial equivalence relation} on $T$: a symmetric and transitive relation on $T$, not necessarily reflexive.  Its domain $\{t\in T : t\sim t\}$ is partitioned into countably many \emph{liveness classes} $T/{\sim}$; a transition outside the domain belongs to no class.
\end{enumerate}
A \emph{computation} of $TS$ is a nonempty, potentially infinite sequence of states $r= s_1,s_2,\cdots$; it is a \emph{run} of $TS$ if $s_1=s_0$.  A \emph{prefix} of a computation is a finite initial segment of it.  A computation $r= s_1,s_2,\ldots$ is \emph{safe}, also written $r\subseteq T$, if $s_i\rightarrow s_{i+1}\in T$ for every two consecutive states.  A class $[t]\in T/{\sim}$ is \emph{enabled} in a state $s$ if $s\rightarrow s'\in[t]$ for some $s'\in S$.  A run $r$ is \emph{live} if no class $[t]\in T/{\sim}$ is enabled in every state of some suffix of $r$ with no member of $[t]$ occurring in the suffix.  A run is \emph{correct} if it is safe and live.
\end{definition}

A partial equivalence, rather than a total one, is used so that some transitions may carry no liveness requirement: only transitions in the domain of $\sim$ form classes and thereby incur a liveness obligation, while transitions outside the domain, belonging to no class, may occur in a correct run but are never required to.

\begin{lemma}[Extension~\cite{lewis2026volitional}]\label{lem:extension}
Every finite safe run of a transition system is a prefix of a correct run of it.
\end{lemma}

\begin{definition}[Multiagent Transition System~\cite{lewis2026volitional}]\label{def:dts-cd}
Given agents $P \subset \Pi$ and an arbitrary set $S$ of \emph{states} with a designated \emph{initial state} $s_0\in S$, a \emph{multiagent transition system} over $P$ and $S$ is a transition system $TS= (C,c_0,T,{\sim})$ with \emph{configurations} $C:= S^P$, \emph{initial configuration}  $c_0:= \{s_0\}^P$, \emph{transitions} $T\subseteq C^2$ a set of transactions over $P$ and $S$, and $\sim$ a partial equivalence on $T$.
\end{definition}

Rather than specifying a multiagent transition system over a set of agents $P$ directly, it is specified via machine transactions.  A machine transaction over $Q\subseteq P$ defines a set of multiagent transitions over $P$ in which all members of $P\setminus Q$ are stationary.

\begin{definition}[Transaction Closure~\cite{lewis2026volitional}]\label{def:closure}
Let $P\subset \Pi$, $S$ a set of machine states, and $C:=S^P$.  For any transition or transaction $t = c\to c'$, we write $t_q := c_q\to c'_q$ and say $p$ is \emph{stationary} in $t$ if $c_p = c'_p$.  For a machine transaction $t=(c\rightarrow c')$ over $S$ with participants $Q$, the \emph{$P$-closure of $t$}, $t{\uparrow}P$, is the set of transitions over $P$ and $S$ defined by:
$$
t{\uparrow}P := \begin{cases} \{ t' \in C^2  :
\forall q\in Q.(t_q = t'_q) \wedge \forall p\in P\setminus Q.(p\text{ is stationary in }t')\} & \text{if } Q\subseteq P \\
\emptyset & \text{otherwise}
\end{cases}
$$
If $R$ is a set of machine transactions, each $t\in R$ over some $Q$ and $S$, then the \emph{$P$-closure of $R$}, $R{\uparrow}P$, is the set of transitions over $P$ and $S$ defined by $R{\uparrow}P := \bigcup_{t\in R} t{\uparrow}P$.  Given a relation ${\sim}$ on $R$, its \emph{$P$-closure} ${\sim}{\uparrow}P$ is the relation on $R{\uparrow}P$ with $\hat t \mathrel{({\sim}{\uparrow}P)} \hat t'$ iff $\hat t \in t{\uparrow}P$ and $\hat t' \in t'{\uparrow}P$ for some $t \sim t'$.  A transition over $P$ is \emph{unary} if it lies in the $P$-closure of a unary transaction.
\end{definition}

If distinct transactions in $R$ have disjoint $P$-closures, as when every participant of every transaction in $R$ changes state, then each transition in $R{\uparrow}P$ has a unique inducing transaction, ${\sim}{\uparrow}P$ relates two transitions exactly when their inducing transactions are $\sim$-related, and ${\sim}{\uparrow}P$ is a partial equivalence whenever ${\sim}$ is.

\begin{definition}[Volitional Multiagent Transition System~\cite{lewis2026volitional}]\label{def:vmts}
Given agents $P\subset \Pi$, machine states $S$ with initial state $s_0$, a set $R$ of volition-guarded transactions such that every $(t,Q')\in R$ has the participants of $t$ contained in $P$ and distinct underlying machine transactions have disjoint $P$-closures, and an equivalence $\sim$ on the set $T_R := \{t : (t,Q')\in R\text{ for some }Q'\}$ of underlying machine transactions, the \emph{volitional multiagent transition system induced by $(S,R,{\sim})$ over $P$} is the multiagent transition system $(\calA^P,c_0,T_V,{\sim_V})$ where:
\begin{enumerate}
    \item $\calA := 2^{T_R/\sim} \times S$ is the \emph{agent state space};
    \item $c_0 \in \calA^P$ is the \emph{initial agent configuration}, with $c_0{^v_p}=\emptyset$ and $c_0{^m_p}=s_0$ for every $p\in P$;
    \item $T_V$ consists of all transitions $e\to e'\in (\calA^P)^2$ of one of two forms: (\ia)~a \emph{volition change} of some $p\in P$---$e^v_p, e'{^v_p}\subseteq (T_R/{\sim})_p$ and $e^v_p \ne e'{^v_p}$, $e^m_p = e'{^m_p}$, and $e_r = e'_r$ for every $r\in P\setminus\{p\}$; or (\ib)~a \emph{volitional machine transaction} induced by some volition-guarded transaction $(t,Q')\in R$ per Definition~\ref{def:vmat}(2);
    \item $\sim_V$ is the restriction of ${\sim}{\uparrow}P$ (Definition~\ref{def:closure}) to the volitional machine transactions, relating two of them whenever their inducing machine transactions are $\sim$-equivalent, and leaves every volition-change transition outside its domain, so volition changes belong to no class and carry no liveness obligation.
\end{enumerate}
\end{definition}

\begin{definition}[Enabled~\cite{lewis2026volitional}]\label{def:enabled}
Given a set of volition-guarded transactions, each $(t,Q')$ with $t = d \rightarrow d'$ a machine transaction over some $Q\subseteq P$ and $S$ with guards $Q' \subseteq Q$, and an equivalence $\sim$ on machine transactions: the volition-guarded transaction $(t,Q')$ is \emph{machine-enabled} in agent configuration $c$ over $P$ if $c^m_p = d_p$ for every $p \in Q$, and \emph{enabled} if in addition $[t] \in c^v_q$ for every $q \in Q'$.  An equivalence class $[t]$ is \emph{enabled} in $c$ if some volition-guarded transaction $(t',Q')$ with $t' \in [t]$ is enabled in $c$.
\end{definition}

A volition-guarded transaction with an empty guard requires no volitions and is enabled whenever its machine precondition is met.  A class of volitional machine transactions is enabled at a configuration in the sense of Definition~\ref{def:ts} exactly when some volition-guarded transaction inducing it is enabled in the above sense, and this determines which runs are live and correct.  Volition changes belong to no class and so impose no liveness obligation; personal choices remain free.

\subsection{Protocols, Grassroots Protocols and Coalescence}\label{app:protocols}

A protocol is a family of multiagent transition systems, one for each set of agents $P\subset \Pi$, which share an underlying set of local machine states $\calS$ with a designated initial state $s_0$.  A \emph{local-states function} maps every set of agents $P \subset \Pi$ to a set of local machine states $S(P)\subset \calS$ that includes $s_0$ and satisfies $P\subset P' \subset \Pi \implies S(P) \subset S(P')$.

\begin{definition}[Protocol~\cite{lewis2026volitional}]\label{def:family}
A \emph{protocol} $\calF$ over a local-states function $S$ is a family of multiagent transition systems that has exactly one transition system $\calF(P) = (C(P),c_0(P),T(P),{\sim(P)})$ for every $P \subset \Pi$, with \emph{agent states} $\calA(P)$, configurations $C(P) := \calA(P)^P$, initial configuration $c_0(P)\in C(P)$, and partial equivalence $\sim(P)$ on $T(P)$ determined by the protocol, such that $P\subseteq P' \subset \Pi$ implies $\calA(P)\subseteq \calA(P')$ and $c_0(P)_p = c_0(P')_p$ for every $p\in P$.
\end{definition}

In a grassroots protocol two disjoint groups of agents can each operate independently --- their interleaved correct runs are correct runs of the combined system --- yet the combined system offers behaviours that neither group could produce on its own.

\begin{definition}[Interleaving~\cite{lewis2026volitional}]\label{def:interleaving}
Let $P, P' \subset \Pi$ be disjoint nonempty sets of agents, $r = c_0, c_1, \ldots$ a run of $\calF(P)$, and $r' = d_0, d_1, \ldots$ a run of $\calF(P')$.  An \emph{interleaving} of $r$ and $r'$ is a sequence $e_0, e_1, \ldots$ of configurations in $C(P \cup P')$ for which there exist non-decreasing sequences of indices $(i_k)_{k \geq 0}$ and $(j_k)_{k \geq 0}$ with $i_0 = j_0 = 0$ such that for every $k \geq 0$:
\begin{enumerate}
\item $(e_k)_p = (c_{i_k})_p$ for every $p \in P$,
\item $(e_k)_q = (d_{j_k})_q$ for every $q \in P'$,
\item if $e_{k+1}$ exists, then exactly one of:
  (\ia) $i_{k+1} = i_k + 1$ and $j_{k+1} = j_k$, a $P$-step, or
  (\ib) $i_{k+1} = i_k$ and $j_{k+1} = j_k + 1$, a $P'$-step.
\end{enumerate}
Moreover, if $r$ is finite of length $n$ then $i_k = n$ for some $k$, and if $r$ is infinite then for every $m \ge 0$ there is a $k$ with $i_k = m$; likewise for $r'$ and $(j_k)$.
\end{definition}

An interleaving is well defined: by Definition~\ref{def:family}, $\calA(P) \subseteq \calA(P\cup P')$ and $\calA(P') \subseteq \calA(P\cup P')$, so each $e_k$, with $p$-components in $\calA(P)$ and $q$-components in $\calA(P')$, is a configuration in $C(P\cup P')$.  Also $e_0 = c_0(P\cup P')$, by the agreement of initial configurations across $\calF(P)$, $\calF(P')$ and $\calF(P\cup P')$.

Interaction, interactive run and first interaction are Definition~\ref{def:first-interaction} of Section~\ref{sec:transactions}.

\begin{definition}[Oblivious, Interactive, Grassroots~\cite{lewis2026volitional}]\label{def:grassroots}
A protocol $\calF$ is:
\begin{enumerate}
    \item \emph{oblivious} if for every disjoint nonempty $P, P' \subset \Pi$, every interleaving of a correct run of $\calF(P)$ and a correct run of $\calF(P')$ is a correct run of $\calF(P\cup P')$.
    \item \emph{interactive} if for every $Q\subset\Pi$ and disjoint nonempty $P, P' \subseteq Q$, every prefix of a correct run of $\calF(Q)$ containing no interaction between $P$ and $P'$ is a prefix of a correct run of $\calF(Q)$ that contains one.
    \item \emph{grassroots} if it is oblivious and interactive.
\end{enumerate}
\end{definition}

Being oblivious means two disjoint groups coexist without interference.  Being interactive means two groups that have not yet interacted can still do so, and the condition is required at every such point and not at the initial configuration only, so a protocol whose groups can reach a state from which they can never couple is not interactive.

The interaction graph, instances and coalescence are Definition~\ref{def:coalescence} of Section~\ref{sec:transactions}.

\begin{proposition}[Coalescence~\cite{lewis2026volitional}]\label{prop:coalescence}
Let $\calF$ be a grassroots protocol, $Q\subset\Pi$, and $P, P'\subseteq Q$ disjoint instances in a prefix $r$ of a correct run of $\calF(Q)$.  Then $r$ is a prefix of a correct run of $\calF(Q)$ in which $P$ and $P'$ coalesce.
\end{proposition}

\subsection{Volitional Transactions-Based Protocols}\label{app:volitional-protocols}

A volitional transactions-based protocol assigns to each set of agents the volitional multiagent transition system induced by its volition-guarded transactions.

\begin{definition}[Transactions Over a Local-States Function~\cite{lewis2026volitional}]\label{def:tblsf}
Let $S$ be a local-states function.  A set of transactions $R$ is \emph{over $S$} if every transaction $t\in R$ is a multiagent transition over $Q$ and $S(P')$ for some $Q \subseteq P'\subset \Pi$.  Given such a set $R$ and $P\subset \Pi$, $R(P) := \{ t\in R : t \text{ is over } Q \text{ and } S(P'), Q \subseteq P'\subseteq P\}$.
\end{definition}

\begin{definition}[Volitional Transactions-Based Protocol~\cite{lewis2026volitional}]\label{def:protocol-transactions}
Let $S$ be a local-states function and $R$ a set of volition-guarded transactions over $S$ with equivalence $\sim$.  The \emph{protocol $\calF$ over $R$, $S$, and $\sim$} assigns to each set of agents $P\subset \Pi$ the volitional multiagent transition system $\calF(P)$ induced by $(S(P),R(P),{\sim})$ over $P$ (Definition~\ref{def:vmts}).  In particular, $C(P) = \calA(P)^P$ with agent state space $\calA(P) := 2^{T_{R(P)}/\sim} \times S(P)$, monotone in $P$ as Definition~\ref{def:family} requires, and $c_0(P)$ has $p$-component $(\emptyset,s_0)$ for every $p\in P$.
\end{definition}

Throughout, $\sim$ is an equivalence on $T_R$, the machine transactions underlying the whole set $R$; for $T'\subseteq T_R$ we write $T'/{\sim}$ for the set of $\sim$-classes with a representative in $T'$, and $[t]$ for the class of $t$ in $T_R/{\sim}$.  Hence $P\subseteq P'$ implies $T_{R(P)}/{\sim}\subseteq T_{R(P')}/{\sim}$.

The three results below are used in the proofs of Section~\ref{sec:polite}.  The first is the invariant that in a group's own system an agent wills only transactions internal to the group, so the guard of a transaction that reaches another group cannot will it while the group runs by itself.

\begin{lemma}[Volitional Containment~\cite{lewis2026volitional}]\label{lem:volitional-containment}
Let $\calF$ be a volitional transactions-based protocol over a set of volition-guarded transactions $R$ with equivalence $\sim$.  For every $P\subset \Pi$, every configuration $c$ of $\calF(P)$, and every $p\in P$: $c^v_p \subseteq T_{R(P)}/\!\sim$, where $T_{R(P)} := \{t : (t,Q')\in R(P)\text{ for some }Q'\}$.
\end{lemma}

\begin{lemma}[Interleaving Safety~\cite{lewis2026volitional}]\label{lem:interleaving-safety}
Let $\calF$ be a volitional transactions-based protocol and $P, P'\subset \Pi$ disjoint and nonempty.  Every finite prefix of an interleaving of a correct run of $\calF(P)$ and a correct run of $\calF(P')$ is a finite safe run of $\calF(P\cup P')$.
\end{lemma}

\begin{proposition}[\cite{lewis2026volitional}]\label{prop:volitional-oblivious}
A volitional transactions-based protocol is oblivious provided that for every disjoint nonempty $P, P' \subset \Pi$, no equivalence class whose transactions have participants spanning both $P$ and $P'$ is ever enabled in any interleaving of correct runs of $\calF(P)$ and $\calF(P')$.
\end{proposition}

In a run of $\calF(P\cup P')$ a member may will a transaction reaching the other group, which a group's own run forbids; the volitional notion constrains the resulting coupling.

\begin{definition}[Volitionally Grassroots~\cite{lewis2026volitional}]\label{def:volitionally-grassroots}
A volitional transactions-based protocol $\calF$ is \emph{volitionally grassroots} if it is grassroots and, for every disjoint nonempty $P, P'\subset\Pi$ and every safe run of $\calF(P\cup P')$ interactive between $P$ and $P'$, the first interaction of $P$ and $P'$ is induced by a volition-guarded transaction $(t,Q')$ with $Q'\cap P\ne\emptyset$ and $Q'\cap P'\ne\emptyset$.
\end{definition}